\documentclass[10pt, letterpaper, twocolumn]{IEEEtran}
\IEEEoverridecommandlockouts     

\usepackage{amsmath,amssymb}   
\usepackage{amsthm}
\usepackage{color}
\usepackage{latexsym} 
\usepackage[normalem]{ulem}   
\usepackage{cite}
\usepackage{graphicx}
\usepackage{multirow}
\usepackage[english]{babel}    

\usepackage{color}

\usepackage{hhline}
\usepackage{pdfpages}
\usepackage{soul}
\usepackage[small]{caption}
\usepackage{url}
\usepackage{amsmath} 
\usepackage{amssymb}  
\usepackage[ruled,vlined]{algorithm2e}

\newtheorem{proposition}{Proposition}

\newtheorem{definition}{Definition}
\newtheorem{example}{Example}
\newtheorem{remark}{Remark}
\newtheorem{problem}{Problem}
\newtheorem{assumption}{Assumption}
\title{Multi-Robot Data Gathering Under Buffer \\ Constraints  and Intermittent Communication}

\author{
Meng Guo~\IEEEmembership{Student Member,~IEEE} and Michael M. Zavlanos~~\IEEEmembership{Member,~IEEE} %
\thanks{The authors are with the Department of Mechanical Engineering and Materials Science, Duke University, Durham, NC 27708 USA. Emails: {\tt\small meng.guo, michael.zavlanos@duke.edu}. This work is supported in part by the NSF awards CNS \#1261828 and CNS \#1302284.
}
}

\begin{document}
\maketitle \thispagestyle{empty} \pagestyle{empty}
\begin{abstract}
 We consider a team of heterogeneous robots which are deployed within a common workspace to gather different types of data. The robots have different roles due to different capabilities: some gather data from the workspace (source robots) and others receive data from source robots and upload them to a data center (relay robots).
The data-gathering tasks are specified locally to each source robot as high-level Linear Temporal Logic (LTL) formulas, that capture the different types of data that need to be gathered at different regions of interest.
All robots have a limited buffer to store the data.
 Thus the data gathered by source robots should be transferred to relay robots before their buffers overflow, respecting at the same time limited communication range for all robots.
  The main contribution of this work is a distributed motion coordination and intermittent communication scheme that guarantees the satisfaction of all local tasks, while obeying the above constraints.
  The robot motion and inter-robot communication  are closely coupled and coordinated during run time by scheduling intermittent meeting events to facilitate the local plan execution.
We present both numerical simulations  and experimental studies to demonstrate the advantages of the proposed method over existing approaches that predominantly require all-time network connectivity.  
\end{abstract}
\section{Introduction}\label{sec:introduction}
\PARstart{M}{any} applications involve robots that are deployed in a workspace to gather different types of data and upload them to a data center for processing.
For instance, teams of unmanned ground vehicles (UGV) can monitor the temperature, humidity, and stand density in large forests or teams of unmanned aerial vehicles (UAV) can monitor the behavior of animal flocks and growth of the crops in farmlands~\cite{dunbabin2012robots}.
{Due to heterogeneous sensing and motion capabilities, the robots in these applications can gather different types of data in different regions within the workspace.} 
Thus the robots can  be assigned local data-gathering tasks that vary across the team~\cite{dunbabin2012robots}.
In this work, we employ Linear Temporal Logic (LTL) as the formal language to describe complex high-level tasks beyond the classic point-to-point navigation.
A LTL task formula is usually specified with respect to an abstraction of the robot motion~\cite{bhatia2010sampling, ulusoy2013optimality}. 
Then a high-level discrete plan is found using off-the-shelf model-checking algorithms~\cite{baier2008principles}, and is  executed through low-level continuous controllers~\cite{fainekos2009temporal}. This framework can be extended to allow for both robot motion and actions in the task specification~\cite{guo2017task}.

The above framework has also been applied to multi-robot systems either in a~\emph{top-down} approach where a global LTL task formula is assigned to the whole team of robots~\cite{chen2012formal, fainekos2006translating, kloetzer2011multi, ulusoy2013optimality}, or in a~\emph{bottom-up} manner where an individual LTL task formula is assigned locally to each robot~\cite{guo2015multi, tumova2014receding}.
Here, we favor the latter formalism as it provides a more natural framework to model independent temporal tasks within large teams of robots that have heterogeneous capabilities.
Specifically, we consider two types of robots: source robots that are assigned local tasks to gather different types of data in different regions in the workspace, and relay robots that receive data from source robots and upload them directly to a data center.
All robots have a limited buffer to store the data. Thus the data gathered by source robots should be transferred to relay robots before the buffers overflow.
Moreover, all robots have a limited communication range, so that they can only communicate when they are sufficiently close to each other.

Communication in the field of mobile robotics has typically relied on constructs from graph theory, with line-of-sight models~\cite{arkin2002line, esposito2006maintaining} and proximity graphs~\cite{ji2007distributed, zavlanos2007flocking,Derb11decentralized, schuresko2009distributed,zavlanos2008distributed, guo2014controlling} gaining the most popularity. In most of these problems, the property of interest is connectivity of the communication network as this allows reliable delivery of information between any pair of robots. Approaches that ensure connectivity for all time either maintain all initial communication links between the robots provided that the initial communication network is connected~\cite{ji2007distributed, guo2014controlling,zavlanos2007potential,guo2015communication},
or allow for addition and removal of communication links while ensuring that the connectivity requirement is not violated~\cite{guo2016hybrid, zavlanos2007flocking, Derb11decentralized, schuresko2009distributed, zavlanos2008distributed, zavlanos2011graph}.
Realistic communication models have recently been proposed in \cite{yan2012robotic,le2012adaptive,zavlanos2013network} that take into account path loss, shadowing, and multi-path fading.
The above approaches enforce all-time connectivity thus are rather restrictive. 
Intermittent communication frameworks, on the other hand, allow the robots to occasionally disconnect from the team and accomplish their tasks free of communication constraints. 
Intermittent communication in multi-agent systems has been studied in consensus problems~\cite{wen2014distributed}, coverage problems~\cite{wang2010awareness}, and in delay-tolerant networks~\cite{daly2007social,jones2007practical}. The common assumption in these works is that the communication network is connected infinitely often. In our recent work~\cite{kantaros2016distributed,kantaros16distributed,kantaros2016simultaneous}, we proposed an intermittent connectivity control strategy that ensures the whole team is connected infinitely often for coverage and path optimization problems. 
However, local high-level temporal tasks are not considered there nor is a model of inter-robot data transfer.

The constraint of {limited buffer size} is of practical importance especially for time-critical data-gathering applications and {for local temporal tasks that require an \emph{infinite} sequence of data-gathering actions.}
The work in~\cite{smith2011optimal} considers a single robot transferring data between locations. The proposed approach minimizes the time interval between two consecutive data-uploading time instants. But it does not explicitly model the evolution of the robot's buffer or the inter-robot communication. 
Similar buffer constraints are considered in~\cite{yamauchi1998frontier} for  multi-robot frontier-based exploration. However locally-assigned data-gathering tasks described by LTL formulas are not considered there, nor are communication constraints.  
Another related area is temporal logic task planning under resource constraints. 
The work in~\cite{leahy2016provably} considers a global surveillance task performed by multiple aerial vehicles subject to battery charging constraints.
The multi-vehicle routing problem considered in~\cite{karaman2011linear} proposes a solution based on Mixed-Integer Linear Programming (MILP), which can potentially be extended to include resource constraints.

The main contribution of this work lies in the development of an online distributed framework that jointly controls local data-gathering tasks and data transfer communication events, so that the buffers at every robot never overflow.
The proposed framework guarantees the satisfaction of all local tasks specified as LTL formulas, without imposing all-time connectivity on the communication network. 
The efficiency of the proposed framework compared to a centralized approach and two static approaches is demonstrated via numerical simulations and experimental studies. To the best of our knowledge, this is the first distributed data-gathering framework under intermittent communication that is also online. 
This work is built on preliminary results presented in~\cite{guo2017distributed}. Compared to~\cite{guo2017distributed}, the real-time control and coordination algorithm presented here is more efficient as it allows the relay robots to swap meeting events in order to faster service the source robots, while it also accounts for robot failures, dynamic robot membership, and fixed data centers.
Furthermore, more extensive numerical simulations are presented, as well as experimental results showing the capabilities of our method.

The rest of the paper is organized as follows: Section~\ref{sec:ltl}  introduces some preliminaries on LTL and B\"uchi {Automata}. Section~\ref{sec:formulation} formulates the problem. Section~\ref{sec:solution} discusses the proposed dynamic approach to joint data-gathering and intermittent communication control. Numerical simulations and experiment studies are shown in Sections~\ref{sec:simulate} and~\ref{sec:exp}, respectively. We conclude in Section~\ref{sec:future}.

\section{Preliminaries on LTL}\label{sec:ltl}

 {Atomic propositions} are Boolean variables that can be either true or false.  The  ingredients of an LTL formula are a set of atomic propositions $AP$ and several boolean and temporal operators, with the following syntax~\cite{baier2008principles}:
$$
\varphi::=\top \;|\; p  \;|\; \varphi_1 \wedge \varphi_2  \;|\; \neg \varphi  \;|\; \bigcirc \varphi  \;|\;  \varphi_1\, \textsf{U}\, \varphi_2,
$$
where $\top\triangleq \texttt{True}$, $p \in AP$ and $\bigcirc$ (\emph{next}), $\textsf{U}$ (\emph{until}), $\bot\triangleq \neg \top$. For brevity, we omit the derivations of other useful operators like $\square$ (\emph{always}), $\Diamond$ (\emph{eventually}), $\Rightarrow$ (\emph{implication}).
The semantics of LTL is defined over the infinite words over~$2^{AP}$. Intuitively, $\sigma \in AP$ is satisfied on a word $w = w(1)w(2)w(3)\ldots\in (2^{AP})^{\omega}$ if it holds at $w(1)$, i.e., if $\sigma \in w(1)$. Formula $\bigcirc \, \varphi$ holds true if $\varphi$ is satisfied on the word suffix that begins in the next position $w(2)$, whereas $\varphi_1 \, \textsf{U}\, \varphi_2$ states that $\varphi_1$ has to remain true until $\varphi_2$ becomes true. Finally, $\Diamond \, \varphi$ and $\square \, \varphi$ are true if $\varphi$ holds on $w$ eventually and always, respectively. We refer the readers to Chapter~5 of~\cite{baier2008principles} for the full definition of LTL syntax and semantics.

The language of words that satisfy an LTL formula~$\varphi$ over $AP$ can be fully captured through~\cite{baier2008principles} a Nondeterministic B\"uchi automaton (NBA) $\mathcal{A}_{\varphi} $, defined as $\mathcal{A}_{\varphi}=(Q, \,2^{AP},\, \delta,\, Q_0,\, {F})$,
where $Q$ is a  set of states; $2^{AP}$ is the set of all allowed alphabets; $\delta\subseteq Q\times 2^{AP} \times {Q}$ is a transition relation; $Q_0, \, {F}\subseteq Q$ are the set of initial and accepting states, respectively.
The process of constructing~$\mathcal{A}_{\varphi}$ can be done in time and space $2^{\mathcal{O}(|\varphi|)}$, where $|\varphi|$ is the length of $\varphi$~\cite{baier2008principles}. There are fast translation tools~\cite{gastin2001fast, package} to obtain~$\mathcal{A}_{\varphi}$  given~$\varphi$.


\section{Problem Formulation}\label{sec:formulation}

\subsection{Robot Model}\label{sec:model}
Consider a team of~$N$ dynamical robots where each robot~$i\in \mathcal{N}\triangleq \{1,2,\cdots,N\}$ satisfies the unicycle dynamics:
\begin{equation}\label{eq:unicycle}
    \dot{x}_i = v_i\cos(\theta_i), \quad \dot{y}_i = v_i\sin(\theta_i), \quad \dot{\theta}_i = \omega_i, 
\end{equation}  
where~$p_i(t)=(x_i(t),\,y_i(t))\in \mathbb{R}^2$, {$\theta_i(t)\in (-\pi,\pi]$} are robot~$i$'s position and orientation at time~$t>0$. 
The control inputs are given by~$u_i(t)=(v_i(t),\,\omega_i(t))$ as the linear and angular velocity.
Each robot has a reference linear and angular velocities denoted by~$v_i^{\texttt{ref}}$ and~$\omega_i^{\texttt{ref}}$, {which are used later to estimate the traveling time. }
The workspace is a bounded 2D area~$\mathcal{W}\subset \mathbb{R}^2$, within which there are clusters of obstacles~$\mathcal{O}\subset \mathcal{W}$. 
The free space is denoted by~$\mathcal{F}=\mathcal{W} \backslash \mathcal{O}$. Note that all robots are assumed to be point masses and robot collision is not considered here. 

As mentioned in Section~\ref{sec:introduction}, the robots are categorized into two subgroups, denoted by~$\mathcal{N}^l,\,\mathcal{N}^f\subset \mathcal{N}$ so that~$\mathcal{N}^l\cup \mathcal{N}^f=\mathcal{N}$ and {$\mathcal{N}^l\cap \mathcal{N}^f=\emptyset$}.
Every robot~$i\in \mathcal{N}^f$ is equipped with short-range wireless units and  can {only} send and receive data from other robots $j$ such that~$\|p_i(t)-p_j(t)\|\leq r_i$, where~$r_i>0$ is the communication range, $\forall j\in \mathcal{N}$.
On the other hand, robots in~$\mathcal{N}^l$ are equipped with long-range wireless units and have the extra function to {upload} their stored data to a remote data center.
In other words,  robots in~$\mathcal{N}^f$ are responsible for gathering data about the workspace while robots in~$\mathcal{N}^l$ are in charge of uploading these data to the data center. 
In the sequel, we simply refer to robots in~$\mathcal{N}^f$ as \emph{source robots} and robots in~$\mathcal{N}^l$ as \emph{relay robots}. Note that there is {at least one} source and relay robot, i.e., it holds that~$|\mathcal{N}^f|, |\mathcal{N}^l|\geq 1$.

\begin{remark}\label{remark:immediate}
The fact that the relay robots can upload their stored data immediately to the data center is due to their long-range communication capabilities. This assumption can be  relaxed by choosing several fixed data centers within the workspace {that the relay robots need to visit and upload their data}. More details are provided in Section~\ref{sec:fixed}. \hfill $\blacksquare$
\end{remark}

\subsection{Data-gathering Tasks}\label{sec:task}
Each source robot~$i\in\mathcal{N}^f$ has a {local} data-gathering task associated with different regions in the freespace. Denote by~$\Pi_{i}=\{\pi_{i,1},\,\pi_{i,2},\cdots, \pi_{i,M_i}\}$ the collection of these regions, where~$\pi_{i,\ell}\subset \mathcal{F}$, $\forall \ell=1,2,\cdots,M_i$ and~$M_i>0$. They contain information of interest.
Moreover, there is a set of data-gathering actions that robot~$i$ can perform at these regions, denoted by~$G_i=\{g_{i,0}, g_{i,1},\,g_{i,2},\cdots,g_{i,K_i}\}$, {where~$g_{i,k}$ means that ``type-$k$ data is gathered by robot~$i$''}, $\forall k=1,2,\cdots,K_i$ and~$K_i\geq 1$. By default,~$g_{i,0}$ means doing nothing. 
The time needed to perform each action by robot~$i\in \mathcal{N}^f$ is given by function~$Z_i:G_i\rightarrow \mathbb{R}^+$.

With a slight abuse of notation, we denote the set of robot~$i$'s atomic propositions by~$AP_i=\{\pi_{i,\ell}\wedge g_{i,k},\, \forall \pi_{i,\ell}\in \Pi_i, \forall g_{i,k}\in G_i\}$, where each proposition~$\pi_{i,\ell}\wedge g_{i,k}$ stands for ``robot~$i$ gathers type-$k$ data at region~$\pi_{i,\ell}$''.
Over these atomic propositions, we can specify a high-level data-gathering task, denoted by~$\varphi_i$, following the LTL semantics in~Section~\ref{sec:ltl}. 
Simply speaking,~$\varphi_i$ specifies the desired sequence of data-gathering actions to be performed at certain regions of interest within the workspace.
Note that LTL formulas allow us to specify data-gathering tasks of finite or \emph{infinite} executions. For instance, $\varphi_i=\Diamond((\pi_{i,1}\wedge g_{i,2}) \wedge  \Diamond (\pi_{i,3}\wedge g_{i,4}))$ means that {``robot~$i$ should gather type-2 data at region~1, then type-4 data at region~3''}, or $\varphi_i = \square \Diamond (\pi_{i,6} \wedge g_{i,7}) \wedge \square \Diamond (\pi_{i,7} \wedge g_{i,2})$ means that ``robot~$i$ should \emph{infinitely often} gather type-7 data at region~6 and type-2 data at region~7''.

\begin{remark}\label{remark:relay-robot}
It is worth mentioning that relay robots~$j\in \mathcal{N}^l$ do not have local tasks as their goal is to communicate with source robots and upload data to the data center. This assumption can be relaxed and is part of our future work. \hfill $\blacksquare$
\end{remark}

\subsection{Buffer Size and Communication Constraints}\label{sec:buffer}
Each robot~$i\in\mathcal{N}$ has a \emph{limited buffer} to store data. To simplify the formulation, we quantify the data size into \emph{units}, i.e., robot~$i$ has a buffer to store a maximum number {of}~$\overline{B}_i>0$ units of data, $\forall i\in \mathcal{N}$.
Furthermore, denote by~$b_i(t)\in \mathbb{N}_{\geq 0}$ the number of data units stored in the buffer of any robot~$i\in\mathcal{N}$ at time~$t\geq 0$. Note that~$b_i(0)=0$, $\forall i\in \mathcal{N}$. It {must} hold that~$b_i(t)\leq \overline{B}_i$, $\forall t\geq 0$ such that the buffer of robot~$i$ does not overflow. Whenever robot~$i\in \mathcal{N}^f$ performs a data-gathering action~$g_{i,k}\in G_i$ at time~$t$, $b_i(t)$ changes as follows:
\begin{equation}\label{eq:gather}
b_i(t^+) = b_i(t^-)+ D_i(g_{i,k}),
\end{equation}
where~$D_i:G_i\rightarrow \mathbb{Z}^+$ is the number of data units gathered by performing action~$g_{i,k}\in G_i$; $b_i(t^-)$ and~$b_i(t^+)$ are the number of data units at robot~$i$'s buffer \emph{before} and~\emph{after} the action~$g_{i,k}$ is performed at time~$t\geq 0$.
If~$b_i(t^+)>\overline{B}_i$, then this action~$g_{i,k}$ can \emph{not} be performed as it {will} lead to buffer overflow.
We assume that~$D_i(g_{i,k})\leq \overline{B}_i$, $\forall g_{i,k}\in G_i$, meaning that any action can be performed when the  buffer is zero.

Moreover, any two robots can send and receive data when they are within each other's communication range. In particular, denote by~$\texttt{c}_{ij}: \mathbb{R}\rightarrow \mathbb{Z}^+$ {the data transfer map} from robot~$i$ to robot~$j$ at time~$t>0$. When robot~$i$ transfers~$\texttt{c}_{ij}(t)$ units of data to robot~$j$, their stored data units change by:
\begin{equation}\label{eq:transfer}
    b_i(t^+) = b_i(t^-) - \texttt{c}_{ij}(t) \;\text{and}\; b_j(t^+) = b_j(t^-) + \texttt{c}_{ij}(t),
\end{equation}
where~$b_i(t^+)$ and~$b_i(t^-)$ (or $b_j(t^+)$ and~$b_j(t^-)$) are the stored data units of robot~$i$ (or robot~$j$) before and after the data transfer.
To allow this transfer, two conditions \emph{must} hold: (i)~$c_{ij}(t)\leq b_i(t^-)$ so that robot~$i$ has enough data to transfer; and~(ii) $b_j(t^+)\leq \overline{B}_j$ so that robot~$j$'s buffer does not overflow.

At last, as mentioned earlier, any relay robot~$j\in \mathcal{N}^l$ has an extra function to upload its stored data to the remote data center. Denote by~$\texttt{d}_{j}: \mathbb{R}\rightarrow \mathbb{Z}^+$ the upload function of robot~$j$ at time~$t>0$. When robot~$j$ uploads~$\texttt{d}_{j}(t)$ units of data to the data center, its stored data changes as follows:
\begin{equation}\label{eq:upload}
b_j(t^+)= b_j(t^-)-\texttt{d}_{j}(t),
\end{equation}
where~$b_j(t^+)$ and~$b_j(t^-)$ are defined similarly as before.
Clearly, the uploaded data \emph{must} not be more than the stored data, i.e.,~$\texttt{d}_j(t)\leq b_j(t^-)$ and~$b_j(t^+)\geq 0$.

\subsection{Problem Statement}
{Consider a team of $N$ robots, consisting of $N^f$ source robots and $N^l$ relay robots, that all satisfy the dynamics~\eqref{eq:unicycle}. Each robot~$i\in \mathcal{N}$ has a limited communication rage $r_i$ and a maximum buffer size~$\overline{B}_i$. The robots' onboard buffers change according to~\eqref{eq:gather}-\eqref{eq:upload}.
Furthermore, each source robot~$i\in \mathcal{N}^f$ is assigned a data-gathering task captured by an LTL formula $\varphi_i$ over~$AP_i$.
The problem we address in this paper is (i) the design of motion controllers $u_i$ and action events $D_i$ that satisfy the local tasks~$\varphi_i$, $\forall i\in \mathcal{N}^f$; as well as (ii) the design of sequences of communication events~$\texttt{c}_{ij}$ and $\texttt{d}_j$ that ensure data delivery to the data center without buffer overflow, $\forall i \in \mathcal{N}^f$ and $\forall j\in \mathcal{N}^l$.
Moreover, we seek a solution that is distributed and online, meaning that there is no central coordinator that collects all information and determines the robots' motion and actions.
}


{
Note that, even though data storage is nowadays very cheap and for many practical purposes can be considered unlimited, setting a buffer limit has the advantage that it forces the robots to relay the gathered data to the data center more frequently, before this limit is reached. Thus, buffer constraints can be used to model \emph{urgency} for communication and they are important in case of time critical tasks. 
Such tasks can range from multi-robot surveillance where the urgency to collect information to a data center is related to quicker response times to possible situations, to cooperative transportation where buffer limits can be used to model the loads that the robots can carry and transport to each other. Note that imposing time constraints (compared to buffer constraints that indirectly model urgency to deliver data) would change completely the problem formulation addressed in this paper and is part of our future work. 
}

\begin{remark}\label{remark:not-top-down}
Note that different from ``\emph{top-down}'' approaches~\cite{chen2012formal, kloetzer2011multi}, here the data-gathering tasks are assigned locally to each source robot, not to the whole team. Each source robot does \emph{not} need to know the number of the other source robots or their local tasks. \hfill $\blacksquare$
\end{remark}

\begin{figure}[t]
\centering
\includegraphics[width =0.49\textwidth]{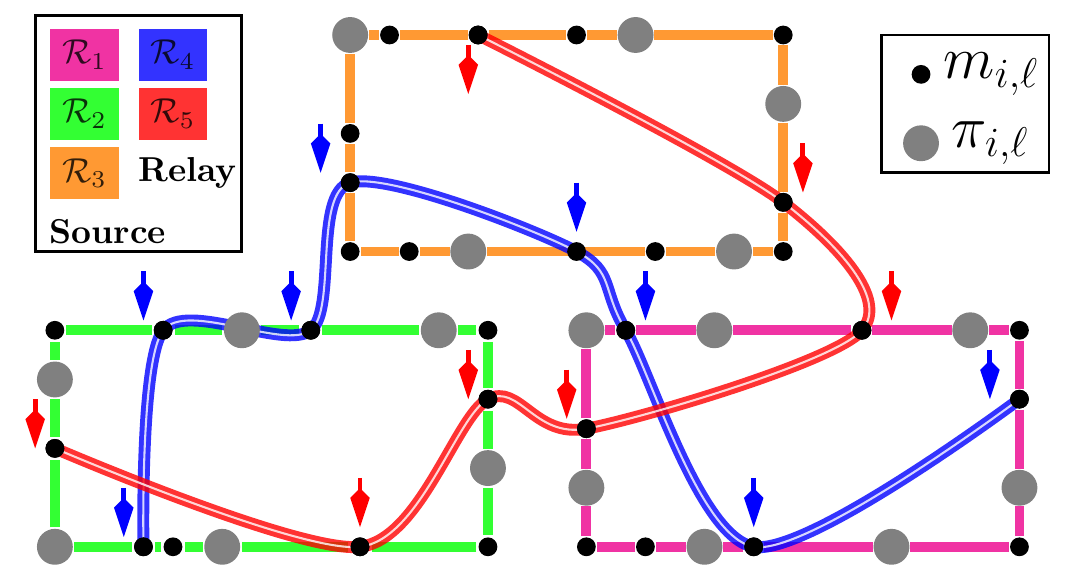}
\caption{Illustration of the proposed solution in Section~\ref{sec:solution}.  Each source robot (in magenta, green, orange) synthesizes its own discrete plan, which includes the regions of interest (in grey-filled circles), the waypoints in between (in black-filled circles) and actions to perform different regions. They coordinate with relay robots (in blue and red) to meet, transfer and upload the gathered data (indicated by blue and red arrows), before their buffers overflow. Note that each source robot can coordinate with multiple relay robots, and vice versa. 
}
\label{fig:idea}
\end{figure}

\section{Centralized Optimal Solution}\label{sec:centralized}
{In this section, we present a centralized solution to the considered problem, which is also the optimal solution. 

The centralized solution consists of three major steps: (i) the construction of a composed transition system for the whole team, which encapsulates all robots' motion and actions (including data gathering, data upload and data exchange). Particularly, the composed FTS is defined as 
\begin{equation}\label{eq:composed}
\mathcal{T}_{\texttt{a}} \triangleq (S_{\texttt{a}},\, \rightarrow_{\texttt{a}},\, S_{\texttt{a},0},\, T_{\texttt{a}},\, AP,\, L_{\texttt{a}}),
\end{equation}
where $S_{\texttt{a}}=(\Pi_1\times \mathcal{B}_1)\times (\Pi_2\times \mathcal{B}_2)\times\cdots (\Pi_N\times \mathcal{B}_N)$ is the set of composed states and $\mathcal{B}_i=\{0,1,\cdots,\overline{B}_i\}$. 
Namely, state $s\in S_{\texttt{a}}$ indicates the position and buffer size of each robot.
The transition relation~$\rightarrow_{\texttt{a}}\subset S_{\texttt{a}}\times S_{\texttt{a}}$ is defined by~$(s,\,s')\in \rightarrow_{\texttt{a}}$ where $s=(\pi_1,b_1)\times\cdots (\pi_N,b_N)$ and $s'=(\pi'_1,b'_1)\times\cdots (\pi'_N,b'_N)$ if robot $i$ is allowed to transition from $\pi_i$ to $\pi_i'$, and if the change in buffer size from $b_i$ to $b_i'$ satisfies both the communication-range constraints and the buffer dynamics defined in~\eqref{eq:gather}-\eqref{eq:upload},  $\forall i\in \mathcal{N}$. The initial state $S_{\texttt{a},0}\in S_{\texttt{a}}$ is given by the initial position and buffer size of the robots. The transition cost $T_{\texttt{a}}:\rightarrow_{\texttt{a}}\rightarrow \mathbb{R}^+$ measures the time that each transition takes. $AP=\cup_{i\in \mathcal{N}^f}AP_i$ is the set of propositions. Lastly, the labeling function $L_{\texttt{a}}:S_{\texttt{a}} \rightarrow 2^{AP}$ reflects the data-gathering actions that have been performed by the source robots at the regions of interest. 
(ii) The conjunction of all source robots' local tasks is defined as $\varphi_{\texttt{a}}=\bigwedge_{i\in \mathcal{N}^f} \varphi_i$, and the corresponding NBA is derived as $\mathcal{A}_{\varphi_{\texttt{a}}}$, as described in Section~\ref{sec:ltl}. (iii) Standard model checking algorithms~\cite{baier2008principles} can be used to search for a lasso-shaped path of~$\mathcal{T}_{\texttt{a}}$ that satisfies~$\varphi_{\texttt{a}}$. These involve constructing the product automaton $\mathcal{P}_{\texttt{a}}$ between $\mathcal{T}_{\texttt{a}}$ and the NBA $\mathcal{A}_{\varphi_{\texttt{a}}}$.
To optimize the total plan cost both in the plan prefix and plan suffix, defined as the accumulated travel time, the synthesis algorithm from our earlier work~\cite{guo2015multi} can be used.

Note that the above solution has two serious drawbacks: first, it is computationally intractable for systems with large numbers of robots and complex tasks, due to the {combinatorial} size of composed system and the double-exponential complexity of the model-checking process~\cite{baier2008principles}.  Second, the derived plan needs to be executed in a \emph{fully-synchronized} way, meaning that the next transition can be taken only if all robots have completed their current transition. Not only does this introduce heavy communication overhead for synchronization but also this all-time synchronization  may be infeasible due to limited communication range considered here. More numerical analyses can be found in Section~\ref{sec:simulate}.

}

\section{Dynamic Data-Gathering and Intermittent Communication Control}\label{sec:solution}
The proposed solution, as shown in Figure~\ref{fig:idea}, consists of three main parts: (i) the workspace abstraction and the synthesis of local discrete plans; (ii) the coordination of meeting events between source and relay robots, including the initial coordination and the real-time coordination; and (iii) the execution of local discrete plans and the data transfer protocol.

\subsection{Local Discrete Plan Synthesis}\label{sec:initial-plan}
Initially at time $t=0$, each source robot~$i\in \mathcal{N}^f$ synthesizes its local discrete plan to satisfy its local task~$\varphi_i$. This plan is given as an infinite sequence of regions to visit and the data-gathering actions to perform at each region.

\subsubsection{Road Map Construction}\label{sec:roadmap}

First, an abstraction of the freespace~$\mathcal{F}$ is constructed as a \emph{roadmap} on which all robots in~$\mathcal{N}$ can move.
\begin{definition}\label{def:road}
The roadmap over the {freespace}~$\mathcal{F}$ is a weighted and undirected graph~$\mathbf{M}=(M,\, H,\, W)$, where~$M$ is the set of waypoints~$m\in \mathbb{R}^2$, $\forall m\in M$, $H\subseteq M\times M$ indicates whether two waypoints are connected, and~$W:H\rightarrow \mathbb{R}^+$ is the Euclidean distance between two waypoints. \hfill $\blacksquare$
\end{definition}

To construct the roadmap $\mathbf{M}$, in this work, we rely on the triangulation algorithm for polygons with holes, see Chapter~6 in~\cite{lavalle2006planning} and the package~``\texttt{poly2tri}'' in~\cite{poly2tri}.
{We omit the algorithmic details due to limited space and refer the interested readers to~\cite{guo2017distributed} and~\cite{Kloetzer15system,de2000computational} for different algorithms.}
An example is shown in Figure~\ref{fig:triangulation}.
This roadmap allows the robots to move among the waypoints without crossing the obstacles.


Using the roadmap~$\mathbf{M}$, we can construct a finite transition system (FTS) to abstract the motion of each source robot~$i\in \mathcal{N}^f$ among its regions of interest within the freespace. 
Denote this motion model by $\mathcal{T}_i = (\Pi_i,\, \rightarrow_i,\, \Pi_{i,0},\, T_i)$,
where~$\Pi_i$ is the set of regions of interest, 
$\rightarrow_i \subseteq \Pi_i \times \Pi_i$ denotes the transition relation,
$\Pi_{i,0}\in \Pi_i$ is the region robot~$i$ starts from initially,
$T_i: \rightarrow_i \to \mathbb{R}^+$ approximates the time each transition takes. Particularly, consider two regions of interest of robot~$i$ denoted by~$\pi_{i,s},\,\pi_{i,f}\in \Pi_i$. Denote by~$m_{i,s},\,m_{i,f}\in M$ the closest waypoints to the center points of~$\pi_{i,s}$ and~$\pi_{i,f}$, respectively.
Then, $(\pi_{i,s},\pi_{i,f})\in \rightarrow_i$ if there exists a path in~$\mathbf{M}$ starting from~$m_{i,s}$ to~$m_{i,f}$ \emph{without} crossing any other waypoint~$m_{i,\ell}\in M$ that belongs to any other region~$\pi_{i,\ell}\in \Pi_i$ with~$\ell \neq s, f$.
Denote the shortest of those paths by~$\Gamma_{i,sf}=m_{i,s}m_{i,s+1}\cdots m_{i,f}$, which can be obtained from a graph search over~$\mathbf{M}$ between~$m_{i,s}$ and~$m_{i,f}$.
Furthermore, 
for each transition~$(\pi_{i,s},\pi_{i,f})\in \rightarrow_i$, the time for robot~$i$ to traverse the associated path~$\Gamma_{i,sf}$ is computed by
\begin{equation}\label{eq:approx-time}
\begin{split}
\textstyle
&T_i(\pi_{i,s},\,\pi_{i,f})=\big{(}\sum\nolimits_{k=s}^{f-1} \|m_{i,k},\,m_{i,k+1}\|\big{)}/v_i^{\text{ref}}\\
&+\big{(}\sum\nolimits_{k=s}^{f-2}\Theta(m_{i,s+1}- m_{i,s},m_{i,s+2}- m_{i,s+1})\big{)}/\omega_i^{\text{ref}},
\end{split}
\end{equation}
where $v_i^{\text{ref}}$, $\omega_i^{\text{ref}}$ are the reference linear and angular velocities as defined in Section~\ref{sec:formulation}, and the function $\Theta:\mathbb{R}^2\times \mathbb{R}^2\to (-\pi,\,\pi]$ computes the angle between two 2D vectors. Note that~$T_i(\cdot)$ is only an estimate of the time it takes for robot~$i$ to travel along each edge.

{Given the motion abstraction~$\mathcal{T}_i$ and the data-gathering actions in~$G_i$, the \emph{complete} robot model can be constructed as shown below; more details can be found in~\cite{guo2017task}.}
\begin{definition}\label{def:robot}
  The complete robot model is the FTS~$\mathcal{R}_i=(\Pi_{i,\mathcal{R}},\, \rightarrow_{i,\mathcal{R}},\, AP_i,\, L_i,\, \Pi_{i,0,\mathcal{R}}, \, T_{i,\mathcal{R}})$, where~$\Pi_{i,\mathcal{R}}=\Pi_i\times G_i$ is the full state; $\rightarrow_{i,\mathcal{R}}\subseteq \Pi_{i,\mathcal{R}}\times \Pi_{i,\mathcal{R}}$ is the transition relation such that~$(\langle \pi_{i,s},g_{i,\ell}\rangle,\langle\pi_{i,f},g_{i,k}\rangle)\in \rightarrow_{i,\mathcal{R}}$ if~(i)~$\langle \pi_{i,s},\pi_{i,f}\rangle\in \rightarrow_i$ and~$g_{i,k}= g_{i,0}$, or (ii)~$\pi_{i,s}=\pi_{i,f}$ and~$g_{i,\ell},g_{i,k}\in G_i$;
  $AP_i$ are the atomic propositions from Section~\ref{sec:task};
the labeling function is defined as~$L_i(\langle \pi_{i,s},g_{i,\ell}\rangle)=\{\pi_{i,s},\, g_{i,\ell}\}$, $\forall \langle \pi_{i,s},\, g_{i,\ell}\rangle \in \Pi_{i,\mathcal{R}}$;
  $\Pi_{i,0,\mathcal{R}}=\langle \Pi_{i,0},\, g_{i,0}\rangle $ is the initial state; and~$T_{i,\mathcal{R}}(\langle\pi_{i,s},g_{i,\ell}\rangle,\langle\pi_{i,f},g_{i,k}\rangle)=T_i(\pi_{i,s},\pi_{i,f}) + Z_i(g_{i,k})$, $\forall (\langle \pi_{i,s},g_{i,\ell}\rangle,\langle\pi_{i,f},g_{i,k}\rangle)\in \rightarrow_{i,\mathcal{R}}$ is the time measure. \hfill $\blacksquare$
\end{definition}

\begin{figure}[t]
\begin{minipage}[t]{0.495\linewidth}
\centering
   \includegraphics[width =1.02\textwidth]{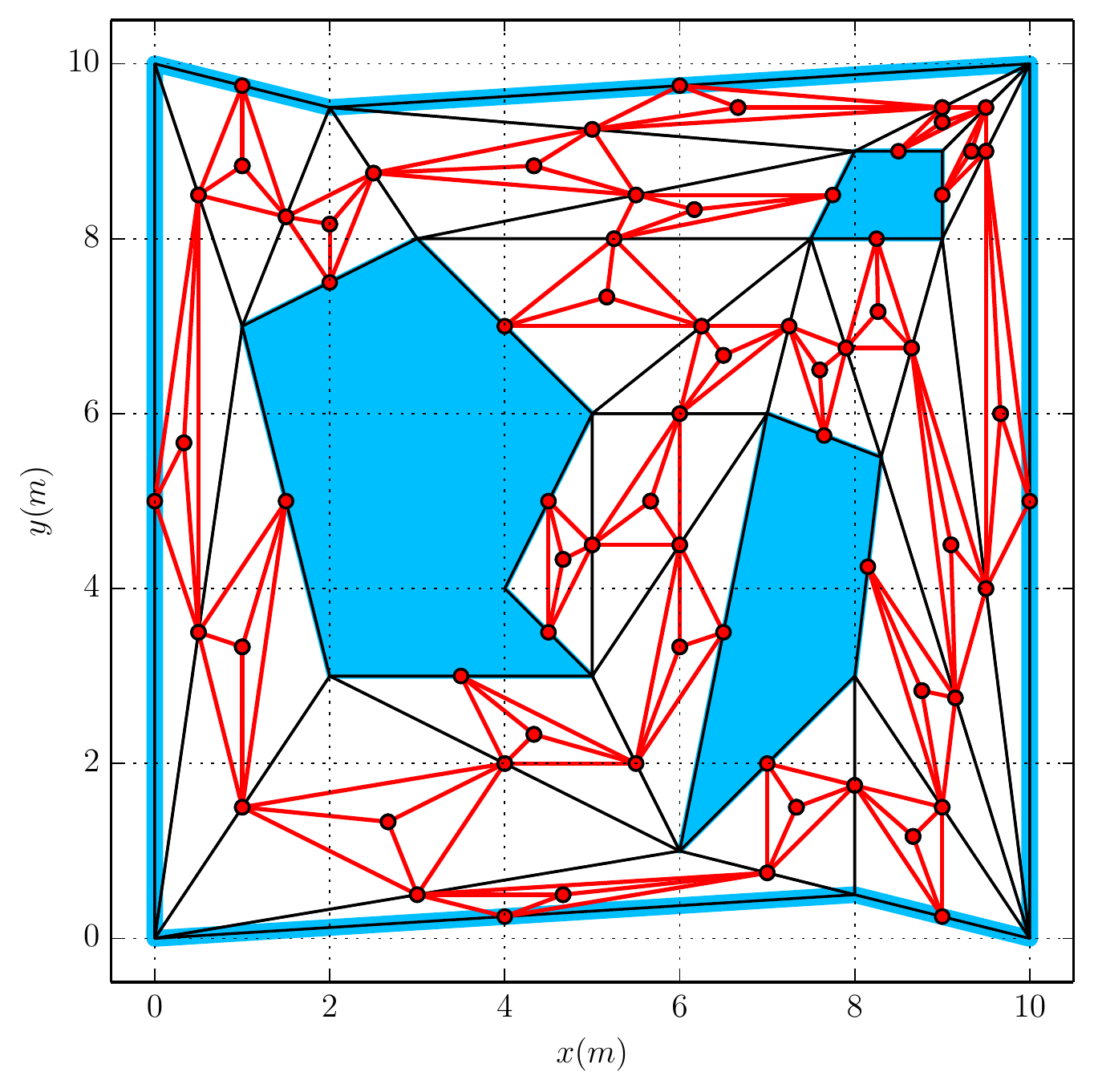}
  \end{minipage}
\begin{minipage}[t]{0.495\linewidth}
\centering
    \includegraphics[width =1.02\textwidth]{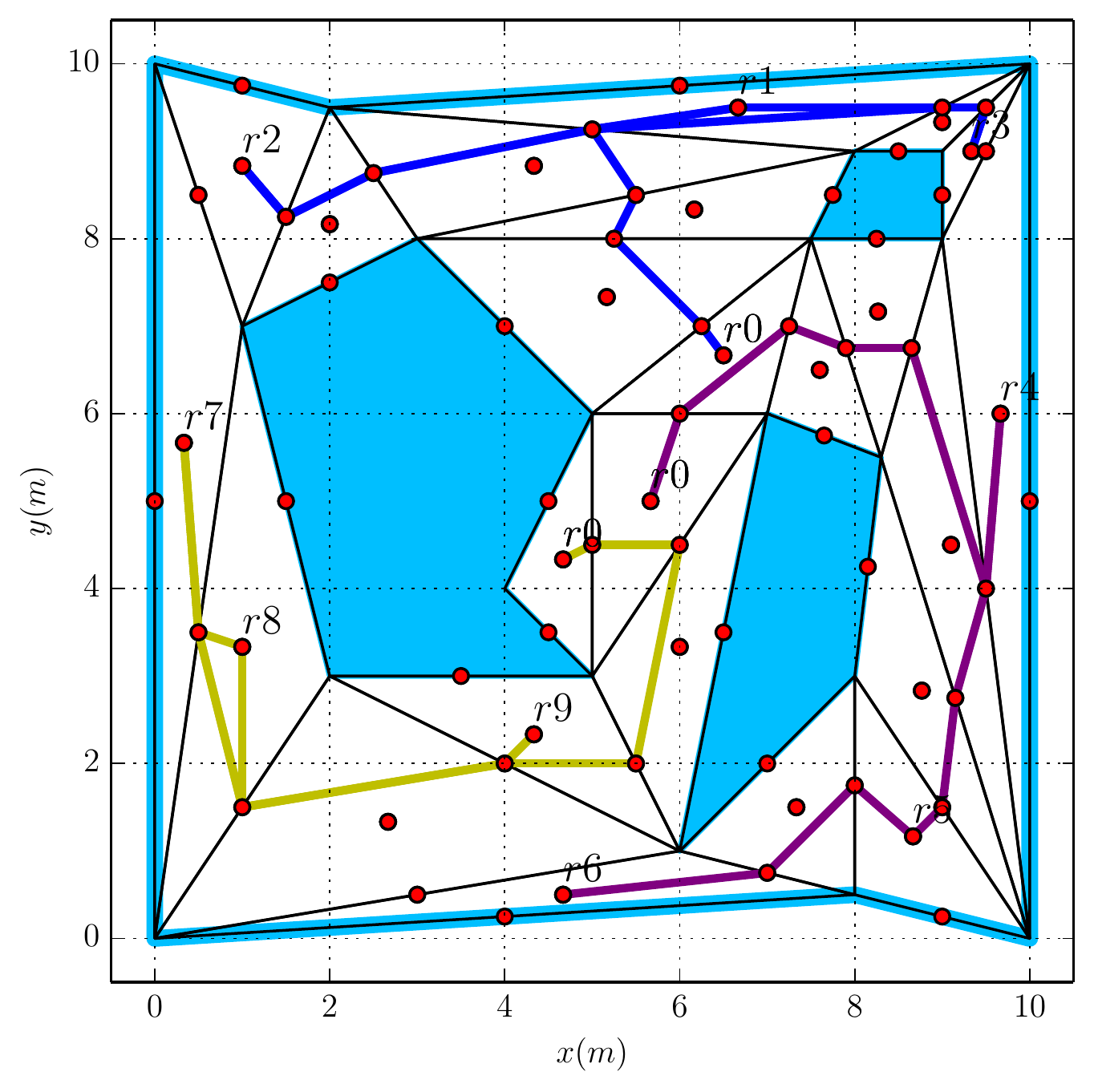}
  \end{minipage}
  \caption{Left: example of the constructed roadmap for the workspace model in Section~\ref{sec:simulate}. Blue areas are boundaries and obstacles. The waypoints and edges are shown by red points and lines; Right: example of the discrete plans for three source robots~$a_0,a_1,a_2$ (in blue, purple, yellow), with regions of interest marked by their labels. }
\label{fig:triangulation}
\end{figure}

\subsubsection{Local Plan Synthesis}\label{sec:synthesis}

The local plan of~robot~$i$, denoted by~$\tau_{i,\mathcal{R}}$, is an infinite path of~$\mathcal{R}_i$ whose trace satisfies its local task~$\varphi_i$.
We rely on the automaton-based model checking algorithm~\cite{baier2008principles,guo2015multi} to synthesize~$\tau_{i,\mathcal{R}}$, {whose description we omit here due to limited space}.
Particularly, the local plan~$\tau_{i,\mathcal{R}}$ has the prefix-suffix structure below and the minimum total cost as the summation of prefix and suffix costs:
\begin{equation}\label{eq:discrete-plan}
\tau_{i,\mathcal{R}} = \pi_{i,\mathcal{R}}^0\,\pi_{i,\mathcal{R}}^1 \cdots \pi_{i,\mathcal{R}}^{k_i-1}\, (\pi_{i,\mathcal{R}}^{k_i}\,\pi_{i,\mathcal{R}}^{k_i+1}\cdots \pi_{i,\mathcal{R}}^{K_i})^\omega,
\end{equation}
where the state~$\pi_{i,\mathcal{R}}^k\in \Pi_{i,\mathcal{R}}$, $\forall k=0,1,\cdots,K_i$ and~$K_i>0$ is the total length, $\pi_{i,\mathcal{R}}^0\,\pi_{i,\mathcal{R}}^1 \cdots \pi_{i,\mathcal{R}}^{k_i-1}$ is the prefix executed only once and $\pi_{i,\mathcal{R}}^{k_i}\,\pi_{i,\mathcal{R}}^{k_i+1}\cdots \pi_{i,\mathcal{R}}^{K_i}$ is the suffix to be repeated infinitely often. $\tau_{i,\mathcal{R}}$ provides an infinite sequence of motion and data-gathering actions to be performed by robot~$i$. Software implementation details can be found in~\cite{guo2015multi, package}.

\begin{example}\label{example:discrete-plan}
Consider the roadmap shown in Figure~\ref{fig:triangulation} within a clustered workspace. Three robots are deployed with different local tasks. For instance, robot~$a_0$ needs to visit~$r_1$, $r_2$ and~$r_3$ in sequence and perform the action~$g_1$ at each region. The resulting discrete plan~$\tau^0_{i,\mathcal{R}}$ is shown in Figure~\ref{fig:triangulation}. \hfill $\blacksquare$
\end{example}

Note that each source robot~$i\in \mathcal{N}^f$ synthesizes~$\tau_{i,\mathcal{R}}$  locally \emph{without} coordination with other robots. 
Thus,  robot~$i$ might not execute~$\tau_{i,\mathcal{R}}$ successfully by itself without the help of relay robots  to transfer data,
due to the infinite sequence of data-gathering actions in~$\tau_{i,\mathcal{R}}$ and its limited buffer size.

\subsection{Coordination of Intermittent Meeting-Events}\label{sec:meet}

To execute the plan of each source robot~$i\in \mathcal{N}^f$, we need to ensure that its stored data is transferred to at least one relay robot~$j\in \mathcal{N}^l$ \emph{before} its buffer overflows.
The main difficulty lies in the limited communication range for both source and relay robots, meaning that both data transfer and coordination are only possible when two robots are within each other's communication range. 
As discussed in Section~\ref{sec:introduction}, instead of imposing all-time connectivity as in most related work~\cite{ji2007distributed,zavlanos2011graph,guo2015communication,guo2016hybrid,Derb11decentralized}, we propose here a distributed online coordination scheme where the communication network is allowed to become disconnected.  

The key idea is to design a method that allows source and relay robots each time they meet (i.e., connect to each other) to negotiate  \emph{when} and \emph{where} they should meet the \emph{next time}, while minimizing the waiting time at the new meeting location. Afterwards, they move independently without communication, until they meet again at the agreed location and time, and the same procedure repeats. 
In the sequel, we present a distributed coordination scheme for both the source robots and relay robots to schedule meeting events, which is based on  online \emph{request and reply} message exchanges, for {four} different scenarios: (i) the initial coordination phase; (ii) the real-time coordination for the next meeting event; (iii) the spontaneous meeting event; and (iv) the swapping of meeting events.

\subsubsection{Initial Coordination}\label{sec:initial-coordination}
Initially at~$t=0$, each source robot needs to coordinate its first meeting event with at least one relay robot. Denote by~$\mathcal{N}_i(t)\subset \mathcal{N}$ the set of robots that robot~$i\in\mathcal{N}$ can communicate with at time~$t\geq 0$, i.e.,~$\mathcal{N}_i(t)=\{j\in \mathcal{N}\,|\,\|p_i(t)-p_j(t)\|\leq r\}$. Then, denote by~$\mathcal{N}^l_{i}(t)=\mathcal{N}_i(0)\cap \mathcal{N}^l$ the set of relay robots that a source robot~$i\in \mathcal{N}^f$ is connected to at time~$t=0$.
{We impose the following assumption on the initial configuration:}
\begin{assumption}\label{assum:initial}
At time~$t=0$, each source robot~$i\in \mathcal{N}^f$ is connected to \emph{at least} one relay robot~$j\in \mathcal{N}^l$:~$\mathcal{N}_i^l(0)\neq \emptyset$. \hfill $\blacksquare$
\end{assumption}

\emph{\textbf{Meeting requests by source robots}}:
To begin with, every source robot~$i\in \mathcal{N}^f$ needs to estimate  where and  when it needs to meet with a relay robot~$j\in \mathcal{N}^l$, given its discrete plan~$\tau_{i,\mathcal{R}}$. We consider the following problem.

\begin{problem}\label{problem:find-state}
For each source robot~$i\in \mathcal{N}^f$, find the \textbf{first} waypoint in the $\tau_{i,\mathcal{R}}$ and the associated time that robot~$i$ needs to meet with a relay robot and transfer data, before robot~$i$'s buffer overflows.    \hfill $\blacksquare$
\end{problem}

To solve Problem~\ref{problem:find-state}, robot~$i\in \mathcal{N}^f$ needs to search through the future sequence of states in~$\tau_{i,\mathcal{R}}$ and determine the~\emph{first} state where the data stored in its buffer \emph{will} exceed its buffer size~$\overline{B}_i$ {if} it has \emph{not} met any relay robot to transfer its data in the meanwhile.
Denote by~$\pi_{i,\mathcal{R}}^{k_e}\in \tau_{i,\mathcal{R}}$ this state and by~$\pi_{i,\mathcal{R}}^{k_t}\in \tau_{i,\mathcal{R}}$ the current state of robot~$i$, where~$k_e>k_t\geq 0$.  Specifically, the index~$k_e$ of $\pi_{i,\mathcal{R}}^{k_e}\in \tau_{i,\mathcal{R}}$ is the index such that
\begin{equation}\label{eq:ke}
\textstyle \sum\nolimits_{k= k_t}^{k_e} D_i(g_{i,\ell_k}){\leq}\overline{B}_i, \;\sum\nolimits_{k= k_t}^{k_e+1} D_i(g_{i,\ell_k}){>} \overline{B}_i, 
\end{equation}
where~$\pi_{i,\mathcal{R}}^k=\langle \pi_{i,s_k}, g_{i,\ell_k}\rangle$, $\forall k_t\leq k\leq k_e$ and~$D_i(g_{i,\ell_k})$ is the number of data units gathered by  action~$g_{i,\ell_k}$ from~\eqref{eq:gather}.
  Thus, the buffer is {less or equal to its full capacity} up to~$\pi_{i,\mathcal{R}}^{k_e}$, but it will overflow at~$\pi_{i,\mathcal{R}}^{k_e+1}$ after performing action~$g_{i,\ell_{k_e+1}}$.

Then, robot~$i$ calculates the route and the associated time to transition from~$\pi_{i,\mathcal{R}}^{k_e}$ to~$\pi_{i,\mathcal{R}}^{k_e+1}$.
Without loss of generality, let~$\pi_{i,\mathcal{R}}^{k_e}|_{\Pi_i}=\pi_{i,s_i}$ and~$\pi_{i,\mathcal{R}}^{k_e+1}|_{\Pi_i}=\pi_{i,f_i}$. 
The shortest path from~$\pi_{i,s_i}$ to~$\pi_{i,f_i}$ is given by~$\Gamma_{i,s_if_i}=m_{i,s_i}m_{i,s_i+1}\cdots m_{i,f_i}$ from Section~\ref{sec:roadmap} and the associated time of reaching each waypoint~$m_{i,s_i}\in \Gamma_{i,s_if_i}$ is denoted by~$t_{i,k_i}\in T_{i,s_if_i}$, where~$T_{i,s_if_i}=t_{i,s_i}t_{i,s_i+1}\cdots t_{i,f_i}$ and~$s_i\leq k_i\leq f_i$. The time sequence~$T_{i,s_if_i}$ is calculated using the reference linear and angular velocities by~\eqref{eq:approx-time}.
As a result, the \emph{request} message from a source robot~$i\in \mathcal{N}^f$ to a relay robot~$j\in \mathcal{N}^l_{i}(0)$ at time~$t=0$, denoted by~$\textbf{Req}_{ij}(0)$, is given by
    \begin{equation}\label{eq:request}
      \textbf{Req}_{ij}(0) = (\Gamma_{i,s_if_i},\,T_{i,s_if_i}),\; \forall j\in \mathcal{N}^l_{i}(0),
      \end{equation}
where~$\Gamma_{i,s_if_i}$ and~$T_{i,s_if_i}$ are defined above. Simply speaking, robot~$i$ is requesting that robot~$j$ should come to meet at any of the waypoints within~$\Gamma_{i,s_if_i}$ at the associated time in~$T_{i,s_if_i}$.

\emph{\textbf{Replies by relay robots}}:
Upon receiving the requests from all source neighbors~$i\in \mathcal{N}_j^f(0)$, where~$\mathcal{N}_j^f(0)\triangleq \mathcal{N}_j(0)\cap \mathcal{N}^f$, each relay robot~$j\in \mathcal{N}^l$ should decide the location and time to meet each source robot~$i\in \mathcal{N}_j^f(0)$ and reply accordingly. 
Denote by~$\textbf{Rep}_{ji}(0)$ the reply message from robot~$j$ to robot~$i$ at time~$t=0$, which has the following structure:
\begin{equation}\label{eq:reply}
  \textbf{Rep}_{ji}(0)=(m_{ji},\,t_{ji}),\; \forall i\in \mathcal{N}_j^f(0)
  \end{equation}
  where~$m_{ji}\in M$ is the waypoint where robots~$i,j$ will meet and~$t_{ji}>t$ is the  time of the meeting event.

Particularly, given the requests~$\textbf{Req}_{ij}(0)=(\Gamma_{i,s_if_i},\,T_{i,s_if_i})$ by~\eqref{eq:request}, $\forall i\in \mathcal{N}_j^f(0)$, we intend to find a path~$\Gamma_j(0)=m_{j,1}m_{j,2}\cdots m_{j,S_j}$, where~$m_{j,s_j}\in M$, $\forall s_j=1,2,\cdots,S_j$ and an associated time sequence~$T_j(0)=t_{j,1}t_{j,2}\cdots t_{j,S_j}$ such that the following {two} conditions hold.
\emph{Condition one}: $\Gamma_j$ should intersect with~$\Gamma_{i,s_if_i}$ exactly once, i.e., there exists exactly one waypoint~$m_{ji}\in \Gamma_{i,s_if_i}$ that~$m_{ji}\in \Gamma_{j}$,~$\forall i\in \mathcal{N}_j^f$. Without loss of generality, let~$m_{ji} = m_{i,k_{ji}}$ where~$s_i\leq k_{ji}\leq f_i$ and $m_{ji} = m_{j, s_{ji}}$ where~$1\leq s_{ji} \leq S_j$; 
\emph{Condition two}: $\Gamma_j$ should minimize the sum of the differences in the predicated meeting time between robot~$j$ and each~$i\in \mathcal{N}_j^f(0)$, i.e.,
$\textstyle{\sum_{i\in \mathcal{N}_j^f}}|t_{i,k_{ji}}-t_{j,s_{ji}}|$, where~$t_{i,k_{ji}}\in T_{i,s_if_i}$ and~$t_{j,s_{ji}}\in T_{j}$ are the corresponding time instances of reaching~$m_{ji}$ in~$\Gamma_{i,s_if_i}(0)$ and~$\Gamma_j(0)$, respectively. Formally, we state the problem below.

\begin{problem}\label{problem:initial-mip}
Given $\textbf{Req}_{ij}(0)=(\Gamma_{i,s_if_i},\,T_{i,s_if_i})$, $\forall i\in \mathcal{N}_j^f(0)$, compute~$\Gamma_j(0)$ such that both conditions above hold. \hfill $\blacksquare$
\end{problem}

Problem~\ref{problem:initial-mip} is closely related to the well-known traveling salesman problem (TSP)~\cite{law2000traveling} but with {three} distinctions: the set of waypoints to be visited is to be determined by the solution; there is no need to return to the starting waypoint; and the cost is defined as the total waiting time over each waypoint instead of the total travel distance.  
The above problem is NP-hard~\cite{lavalle2006planning} as it contains the TSP as a special case. 
{A similar formulation appears in the computer wiring problem as discussed in~\cite{law2000traveling}.
To find the exact solution to Problem~\ref{problem:initial-mip}, we can transform it into a generalized TSP.} 
In particular, let~$\mathcal{N}_{j,+}^{f}= \mathcal{N}_j^f\cup \{j\}\cup\{\nu\}$, where~$\mathcal{N}_j^f$ is the set of source neighbors that robot~$j$ is connected to and~$\nu$ is an artificial node. 
Recall that the requests~$\textbf{Req}_{ij}(0)=(\Gamma_{i,s_if_i},\,T_{i,s_if_i})$ satisfy~$\Gamma_{i,s_if_i}=m_{i,s_i}m_{i,s_i+1}\cdots m_{i,f_i}$ and~$T_{i,s_if_i}=t_{i,s_i}t_{i,s_i+1}\cdots t_{i,f_i}$. 
For ease of notation, let~$\mathcal{I}_{i}^{sf}\triangleq\{s_i,s_i+1,\cdots,f_i\}$.

As mentioned in condition one, $\Gamma_j$ intersects with~$\Gamma_{i,s_if_i}$ exactly once, $\forall i\in \mathcal{N}_j^f$. Let this happen at the $(k_i)_{th}$ element of~$\Gamma_{i,s_if_i}$, where~$k_i\in \mathcal{I}_{i}^{sf}$, $\forall i\in \mathcal{N}_j^f$. 
{Let us define first the set of waypoints}~$\Upsilon = \{m_{j,0},m_{\nu,0},m_{i,k_i}, \forall i\in \mathcal{N}_i^f, \forall k_i\in \mathcal{I}_i^{sf}\}$, which includes all waypoints within each source robot~$i$'s path segment~$\Gamma_{i,s_if_i}$ and~$m_{j,0}=  m_{\nu,0} \triangleq (x_j(0),y_j(0))$. 
{Note that $\nu$ is an artificial node at the end of $\Gamma_j$}.
Also, to simplify the notation, we set~$\mathcal{I}_j^{sf}\triangleq \{k_{j,0}\}$ and~$\mathcal{I}_{\nu}^{sf}\triangleq \{k_{\nu,0}\}$, {where~$k_{j,0}=k_{\nu,0}\triangleq 0$ denote the first and the only element associated with nodes~$j$ and $\nu$, respectively.}
Furthermore, we define a \emph{cost} function $c:\Upsilon\times \Upsilon \rightarrow \mathbb{R}_{\geq 0}$ between any two nodes in~$\Upsilon$ such that: (i) for all~$i,h\in \mathcal{N}_i^f$, it holds that~
\begin{equation}\label{eq:cij}
c_{{k_i}{k_h}}\triangleq |t_{i,{k_i}} + T_j(m_{i,{k_i}}, m_{h,{k_h}}) -t_{h,{k_h}}|,
\end{equation}
where $\forall k_i\in \mathcal{I}_i^{sf}$ and $\forall k_h\in \mathcal{I}_h^{sf}$, where~$t_{i,{k_i}},t_{h,{k_h}}$ are the associated time instants of~$m_{i,{k_i}}, m_{h,{k_h}}$ obtained from~$T_{i,s_if_i}$ and~$T_{h,s_hf_h}$ and {the function~$T_j(\cdot)$ is the time it takes robot~$j$ to travel from~$m_{i,{k_i}}$ to~$m_{h,{k_h}}$, which can be computed similarly to~\eqref{eq:approx-time}};
(ii) for all~$i\in \mathcal{N}_i^f \cup \{j\}$, $c_{{k_i}{k_{\nu}}}=0$, $\forall k_i\in \mathcal{I}_i^{sf}$;
and (iii) for all~$i\in \mathcal{N}_i^f$, $c_{{k_{\nu}}{k_i}}=+\infty$, $\forall k_i\in \mathcal{I}_i^{sf}$, and $c_{{k_{\nu}}{k_j}}=0$. Furthermore, let~$\beta_{{k_i} {k_h}}\in \mathbb{B}$ be a Boolean variable so that~$\beta_{{k_i} {k_h}}=1$ if~$\Gamma_j$ contains a segment from~$m_{i,{k_i}}$ to~$m_{h, {k_h}}$, and is~$0$ otherwise, $\forall k_i\in \mathcal{I}_i^{sf}$, $\forall k_i\in \mathcal{I}_i^{sf}$ and $\forall i,h \in \mathcal{N}_{j,+}^{f}$. Given the above notations, we can formulate the  following integer linear program (ILP) {on the variables~$\{\beta_{k_i k_h}\}$}:
\begin{subequations}\label{eq:ilp}
\begin{align}
    &{\textbf{min}_{\{\beta_{k_i k_h}\}}}\sum_{k_i,\, k_h\in \mathcal{I}_{i,h}^{sf};\, i,h\in \mathcal{N}_{j,+}^{f}} c_{k_i k_h}  \,\cdot \beta_{k_i k_h} \tag{\ref{eq:ilp}}\\
& \textbf{s.t.} 
\quad \sum_{ k_h\in \mathcal{I}_{h}^{sf};\, h\in \mathcal{N}_{j,+}^{f}} \beta_{k_h k_i} = \sum_{h\in \mathcal{N}_{j,+}^{f};\,  k_h\in \mathcal{I}_{h}^{sf}} \beta_{k_i k_h}, \label{eq:ilp-c1}\\
& \qquad \qquad \qquad \qquad \qquad \; \forall k_i\in \mathcal{I}_{i}^{sf},\, \forall i\in   \mathcal{N}_{j,+}^{f}; \nonumber \\
  & \quad  \sum_{k_i,\, k_h\in \mathcal{I}_{i,h}^{sf}; \, h\in \mathcal{N}_{j,+}^{f}} \beta_{k_i k_h} = 1, \quad \forall i\in   \mathcal{N}_{j,+}^{f}; \label{eq:ilp-c2}\\
& \quad \alpha_{k_i} - \alpha_{k_h} + (N_j^f+1)\cdot \beta_{k_i k_h} \leq N_j^f,\label{eq:ilp-c3}\\ 
&\qquad \qquad \qquad \quad \forall k_i, k_{h}\in \mathcal{I}_{i,h}^{sf}; \, \forall i,\,h\in   \mathcal{N}_{j}^{f}\cup\{\nu\};\nonumber 
\end{align}
\end{subequations}
where the notation~$k_i, k_h\in \mathcal{I}_{i,h}^{sf}$ is equivalent to~$k_i\in \mathcal{I}_{i}^{sf}$ and~$k_h\in \mathcal{I}_{h}^{sf}$, similar arguments hold for~$k_i, k_j\in \mathcal{I}_{i,j}^{sf}$; and~$\alpha_{k_i}\in \mathbb{Z}$ is used to avoid the existence of multiple cycles, $\forall k_i\in \mathcal{I}_{i}^{sf}$ and $\forall i \in \mathcal{N}_{j}^f\cup \{\nu\}$.
The first two constraints~\eqref{eq:ilp-c1}-\eqref{eq:ilp-c2} ensure that exactly one element of~$\Gamma_{i,s_if_i}$ is intersected by~$\Gamma_j$,~$\forall i\in \mathcal{N}_{j}^{f}$.
The last constraint~\eqref{eq:ilp-c3} {and the definition of variables $\{\alpha_{k_i}\}$} ensure that all the waypoints~$m_{k_i}$ and $m_{k_h}$ that satisfy~$\beta_{{k_i}{k_h}}$ should belong to \emph{one} big cycle where~$m_{\nu,0}$ is the last waypoint and is connected to~$m_{j,0}$. 
Simply speaking, assume that an additional cycle of waypoints (excluding~$m_{j,0}$) with length~$N_c>0$ appears in~$\Gamma_j$. Summing up the inequalities within~\eqref{eq:ilp-c3} for all waypoints contained in that cycle would yield~$N_c\cdot(N_j^f+1)\leq N_c\cdot N_j^f$, leading to a contradiction. {More details can be found in~\cite{law2000traveling}.}

The ILP problem by~\eqref{eq:ilp} has~$\binom{\hat{N}}{2}$ Boolean variables and~$\hat{N}$ integer variables, where~$\hat{N}=\sum_{i\in \mathcal{N}_{j,+}^f}|\Gamma_{i,s_if_i}|$ is the total number of waypoints and $\binom{\hat{N}}{2}$ is the binomial coefficient. 
Thus the complexity of~\eqref{eq:ilp} is closely related to the number of source robots that each relay robot initially connects to and their request messages. 
{Note that~\eqref{eq:ilp} always has a solution as each relay robot~$j\in \mathcal{N}^l$ can reach any waypoint in~$\mathbf{M}$ (thus any waypoint in $\Gamma_{i,s_if_i}$).}
Lastly, given the solutions $\Gamma_j$ and~$T_j$, the replies~$\textbf{Rep}_{ji}(0)$ can be  derived as:
$m_{ji}=m_{i,k_i}$ and~$t_{ji}=t_{i,k_i}$, $\forall i\in \mathcal{N}_j^f(0)$. {Note that during the transition $(m_{j,s}, m_{j,s+1})\in \Gamma_{j}$, robot~$j$ can intersect with $\Gamma_{i,s_if_i}$ more than once but no data exchange will take place with robot $i$.}

\begin{remark}\label{remark:priority}
{If the waiting time of some source robots are penalized more than other source robots, we can readily incorporate this aspect by imposing static priorities within the source robots in Problem~\ref{problem:initial-mip}, by adding different weights in front of the waiting time $c_{k_ik_h}\cdot \beta_{k_ik_h}$.} \hfill $\blacksquare$
\end{remark}

\emph{\textbf{Confirmation by source robots}}:
Upon receiving the replies~$\textbf{Rep}_{ji}(0)$ from all relay robots~$j\in \mathcal{N}^l_i(0)$, each source robot~$i\in \mathcal{N}^f$  evaluates these replies and sends confirmations back.
In particular, denote by~$\textbf{Conf}_{ij}(0)$ the confirmation message from the source robot~$i$ to robot~$j\in \mathcal{N}^l_i(0)$ at time $0$ so that $\textbf{Conf}_{ij}(0) = \top$ if robot~$i$ confirms the meeting location and time with robot~$j$, while~$\textbf{Conf}_{ij}(0)=\bot$ if robot~$i$ refuses the reply and thus is \emph{not} committed to the meeting event with robot~$j$.
Given the replies~$\textbf{Rep}_{ji}(0)=(m_{ji},\,t_{ji})$,~$\forall j\in \mathcal{N}^l_i(0)$, robot~$i$ chooses the relay robot~$j_i^{\star}\in \mathcal{N}^l_i(0)$ that yields the \emph{minimum}  waiting time for itself at the first meeting event, i.e., 
\begin{equation}\label{eq:min}
j_i^{\star}=\textbf{argmin}_{j\in \mathcal{N}^l_i(0)} |t_{ji}-t_{i,k_{ji}}|,
\end{equation}
where~$s_i\leq k_{ji}<f_i$ satisfies that~$m_{i,k_{ji}}=m_{ji}$. Then,~$\textbf{Conf}_{ij_i^{\star}}(0) = \top$, for~$j_i^{\star}$ above, while~$\textbf{Conf}_{ij_i^{\star}}(0) = \bot$, $\forall j\in \mathcal{N}^l_i(0)$ and $j\neq j_i^{\star}$. Thus robot~$i$ marks~$m_{i,k_{j^\star_i}}$ as the meeting location with robot~$j^\star_i$ at time~$t_{i,k_{j^\star_i}}$.

  On the other hand, after receiving the confirmation messages~$\textbf{Conf}_{ij}(0)$ from source robots~$i\in \mathcal{N}_j^f(0)$, each relay robot~$j\in \mathcal{N}^l$ removes the meeting event with each source robot~$i$ from its path~$\Gamma_{j}(0)$ that was computed by~\eqref{eq:ilp} if the confirmation message from  the source robot~$i$ satisfies~$\textbf{Conf}_{ij}(0)=\bot$, $\forall i \in \mathcal{N}_j^f(0)$.
In other words, each relay robot~$j\in \mathcal{N}^l$ is only committed to meet the source robots that have confirmed the meeting event.

\subsubsection{Coordination for Next Meeting Event}\label{sec:next}
After the initial coordination, robots~$i$ and~$j_i^{\star}$ will meet at the waypoint~$m_{j^\star i}$ at time~$t=t_{j_i^\star i}$, $\forall i \in \mathcal{N}^f$.
For the ease of notation, we replace~$j_i^\star$ by~$j$ in this section.
Then, the data at robot~$i$'s buffer will be transferred to robot~$j$'s buffer and will be uploaded to the data center, see Section~\ref{sec:meet-execute}.
When this happens, the two robots will need to coordinate in order to determine their \emph{next} meeting event  following the procedure described below.

First, robot~$i$ needs to determine again the segment of its future plan when it should meet with a relay robot, before its buffer overflows.
The same equation as in~\eqref{eq:ke} can be applied given that the robot's {current} buffer size is zero and~$\pi_{i,\mathcal{R}}^{k_t}$ is the current state. 
Denote the new request message by~$\textbf{Req}_{ij}(t)=(\Gamma_{i,s_if_i},\,T_{i,s_if_i})$, where~$\Gamma_{i,s_if_i}=m_{i,s_i}\cdots m_{i,f_i}$ and~$T_{i,s_if_i}=t_{i,s_i}\cdots t_{i,f_i}$ are defined analogously as before.
Then, after receiving the request, robot~$j$ needs to reply with its preferred \emph{next}  location and time to meet with robot~$i$, denoted by~$m_{ji}^+$ and~$t_{ji}^+$, respectively. 
Let~$\Gamma_{j}(t)=m_{j,k_{j}}\cdots m_{j,f_{j}}$ be the \emph{remaining} path obtained by~\eqref{eq:ilp} at time~$t$, and the associated sequence of time instants is~$T_{j}(t)=t_{j,k_j}\cdots t_{j,f_{j}}$. 
Thus, the last committed meeting location and time  are given by~$m_{j,f_{j}}$ and $t_{j,f_{j}}$.
Then~$m_{ji}^+$ can be chosen among~$\Gamma_{i,s_if_i}$ such that moving from~$m_{j,f_{j}}$ to~$m_{j i}^+$ yields the minimum waiting time for robot~$i$.
Thus, it holds that $m_{ji}^+=m_{i, s_{j i}^+}$ and~$t_{ji}^+=t_{i, s_{j i}^+}$, where the index~$s_{j i}^+$ satisfies that
\begin{equation}\label{eq:next}
  \begin{split}
    s_{j i}^+=\textbf{argmin}_{s_i\leq s_{j i} \leq f_i }&\|t_{j,f_{j}}-t_{i, s_{j i}}\\
    &+T_{j}(m_{j,f_{j}},
\, m_{i, s_{j i}})\|,\\
    \end{split}
  \end{equation} 
where~$T_{j}(m_{j,f_{j}},\, m_{i, s_{j i}})$ is the time to navigate from waypoint~$m_{j,f_{j}}$ to~$m_{i, s_{j i}}$.
 $s_{j i}^+$ can be found by iterating through all waypoints in~$\Gamma_{i,s_if_i}$ to find the minimum waiting time.
Therefore, the reply message from robot~$j$ to~$i$ is given by~$\textbf{Rep}_{ji}=(m_{ji}^+, \, t_{ji}^+)$.
After receiving the reply message, robot~$i$ will send back the confirmation as~$\textbf{Conf}_{ij}=\top$ and mark~$m^+_{ji}$ as the next meeting location with robot~$j$.
On the other hand, after the confirmation, robot~$j$ will \emph{concatenate} its path~$\Gamma_{j}$ with the shortest path from~$m_{j,f_{j}}$ to~$m_{ji}^+$ within~$\mathbf{M}$ and mark~$m_{ji}^+$ as the next meeting location with robot~$i$.

\subsubsection{Spontaneous Meeting Events}\label{sec:spontaneous}
When there are more than one relay robots in the team, 
it is possible that robot~$i\in \mathcal{N}_i^f$ meets with \emph{another} relay robot~$j'\in \mathcal{N}^l$ on its way to meet the confirmed relay robot~$j_i^{\star}$. We call this situation a \emph{spontaneous} meeting event.
In this case, robot~$i$ transfers the stored data in its buffer to robot~$j'$, and coordinates with~$j'$  for the next meeting event in a similar way as described in Section~\ref{sec:next}, but now robot~$i$ takes into account the fact that it will meet with~$j_i^{\star}$ at~$m_{j_i^{\star}i}^+$ as previously confirmed. 
Thus, the next path segment of~$\Gamma_i$ where robot~$i$ needs to meet with a relay robot should be calculated as in~\eqref{eq:ke} by setting~$\pi_{i,\mathcal{R}}^{k_t}=(m_{j_i^{\star}i}^+,\, g_0)$, i.e., robot~$i$'s buffer is zero {after} meeting robot~$j^\star$ at~$m_{j_i^{\star}i}^+$.
After the coordination with robot $j'$, robot~$i$ continues to meet robot~$j_i^{\star}$.
In this way, a source robot can meet and transfer data through~\emph{all} relay robots it has {met}, instead of being restricted to the relay robot it was connected to initially.
Each time it coordinates with a new relay robot, it takes into account the fact that it will meet with all the relay robots it has committed to and particularly its buffer will be empty after the last meeting event. 


It is crucial that the source robot~$i$ \emph{still} meets its initially confirmed relay robot~$j^\star_i$ (even with an empty buffer), after a spontaneous meeting with another relay robot~$j'\in \mathcal{N}^l$. 
Due to the limited communication range, robot~$i$ can not inform robot~$j^\star_i$ to \emph{cancel} the confirmed next meeting. If robot~$i$ simply skips that meeting, robot~$j^\star_i$ will wait for robot~$i$ at the confirmed region indefinitely, which leads to a deadlock.

\begin{remark}\label{remark:inter-type-A}
Note that source robots  are \emph{not} allowed to transmit data to each other even when they are within the communication range. This assumption can be relaxed and is part of our ongoing work. \hfill $\blacksquare$
\end{remark}

\subsubsection{Relay Robots Swap Meeting Events}\label{sec:swap}

Until now, we have discussed the communication between source and relay robots. In this part, we discuss how relay robots can communicate with each other and swap their committed meeting events with source robots.
Particularly, assume that two relay robots~$j_1, j_2\in \mathcal{N}^l$ meet at time~$t'>0$. The remaining path and the associated time stamps of robot~$j_1$ are given by~$\Gamma_{j_1}(t')=m_{j_1,k_{j_1}}\cdots m_{j_1,f_{j_1}}$ and $T_{j_1}(t')=t_{j_1,k_{j_1}}\cdots t_{j_1,f_{j_1}}$, respectively.
Similarly,~$\Gamma_{j_2}(t')=m_{j_2,k_{j_2}}\cdots m_{j_2,f_{j_2}}$ and $T_{j_2}(t')=t_{j_2,k_{j_2}}\cdots t_{j_2,f_{j_2}}$ for robot~$j_2$.
Our goal is to rearrange the entries in~$\Gamma_{j_1}$ and~$\Gamma_{j_2}$ such that the total waiting time for source robots is further reduced. 

\begin{figure}[t]
  \centering
  \includegraphics[width =0.5\textwidth]{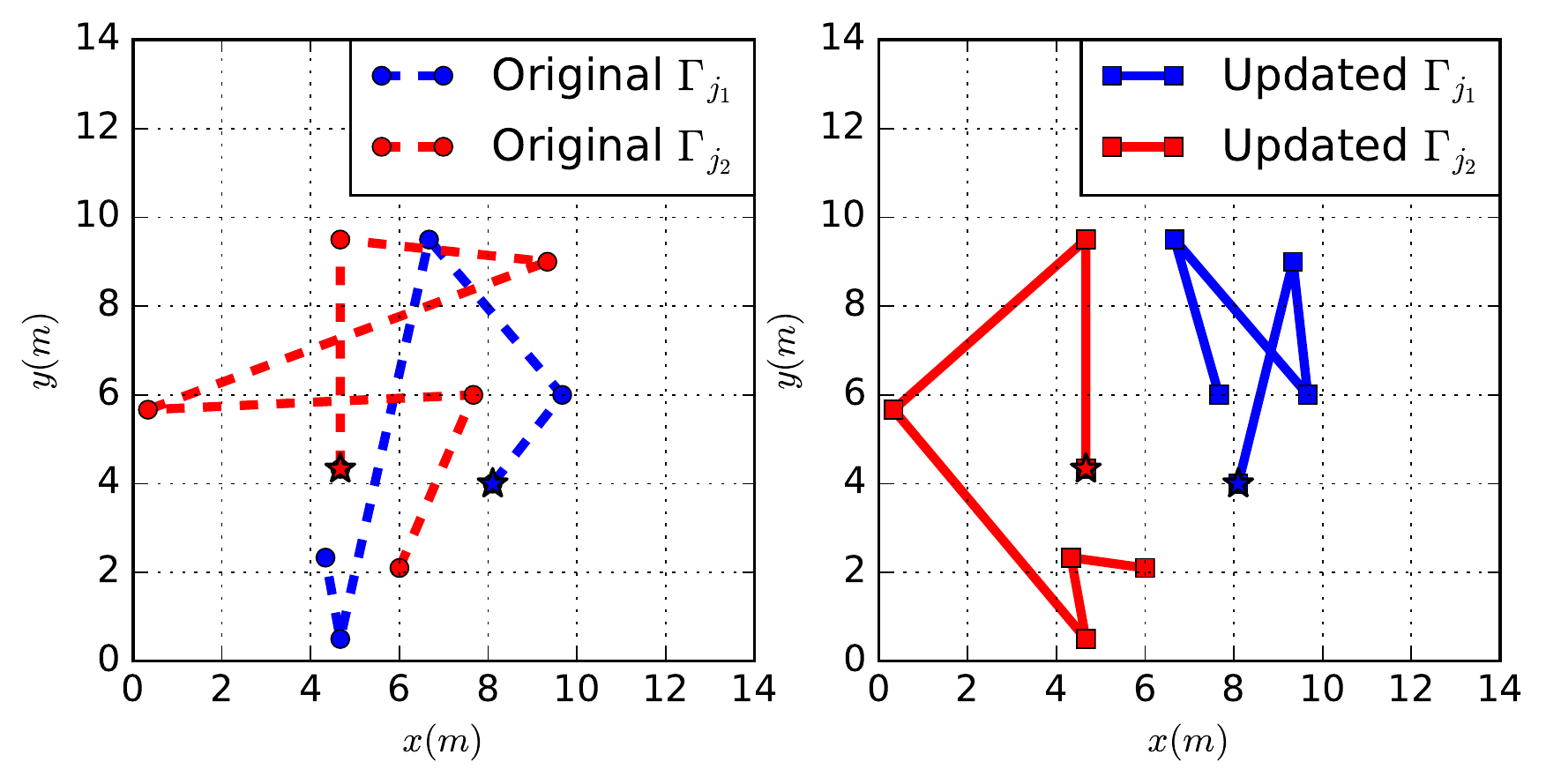}
  \caption{Visualization for Example~\ref{example:swap} of robots~$j_1,j_2$'s paths before and after the swapping algorithm presented in Section~\ref{sec:swap}. Initial position of robots~$j_1$ and~$j_2$ are indicated by filled stars. 
}
  \label{fig:swap}
\end{figure}

Clearly, the optimal way to rearrange~$\Gamma_{j_1}$ and~$\Gamma_{j_1}$ that yields the minimum waiting time is to formulate a integer linear problem similar to~\eqref{eq:ilp}. It can be thought of as a traveling salesman problem with two salesmen. 
Here we propose a greedy algorithm that takes advantage of the \emph{ordered} structure of~$\Gamma_{j_1}$ and~$\Gamma_{j_1}$. 
First, we construct a new sequence of 2-tuples~$\Upsilon = (m_1,t_1)(m_2,t_2)\cdots(m_L,t_L)$, where~$L=|\Gamma_{j_1}|+|\Gamma_{j_2}|$. It holds that~$m_l= \Gamma_{j_1}[l_1]$ and  $t_l= T_{j_1}[l_1]$ with the index~$l_1$ that satisfies~$k_{j_1}\leq l_1 \leq f_{j_1}$, or  $m_l= \Gamma_{j_2}[l_2]$ and  $t_l= T_{j_2}[l_2]$ with the index~$l_2$ that satisfies~$k_{j_2}\leq l_2 \leq f_{j_2}$, $\forall l=1,\cdots,L$. More importantly, $\Upsilon$ is ordered by~$t_1\leq t_2\leq \cdots t_L$, i.e., an increasing time order according to which each waypoint should be visited.
Second, let $\Upsilon_1$ and $\Upsilon_2$ be two subsequences of~$\Upsilon$ that we want to construct. They are initialized by~$\Upsilon_1=(m_{j_1}(t'),t')$  and $\Upsilon_2=(m_{j_2}(t'),t')$, where~$m_{j_1}(t')$ and~$m_{j_2}(t')$ are the waypoints robots~$j_1$ and~$j_2$ are located, respectively. Then we iterate over each entry of~$(m_l, t_l)\in \Upsilon$ and evaluate the waiting time using~\eqref{eq:next} if the paths of robots~$j_1$ or~$j_2$ contain this entry as their last meeting event. If robot~$j_1$ yields a smaller waiting time, we add~$(m_l,t_l)$ to the end of~$\Upsilon_1$; otherwise, 
if robot~$j_2$ yields a smaller waiting time, we add~$(m_l,t_l)$ to the end of~$\Upsilon_2$.
At last,~$\Upsilon_1$ is decomposed into the new~$\Gamma_{j_1}$ and~$T_{j_1}$ for robot~$j_1$, while $\Upsilon_2$ is decomposed into the new~$\Gamma_{j_2}$ and~$T_{j_2}$ for robot~$j_2$. 
Since the above algorithm is greedy, we can compare the total waiting time under the new paths~$\Gamma_{j_1}$ and~$\Gamma_{j_2}$, which is then compared to the original total waiting time. 
If the total waiting time is reduced, the updated~$\Gamma_{j_1}$ and~$\Gamma_{j_2}$ will be used; otherwise, the paths remain unchanged.
In this way, some of the meeting events are swapped between relay robots~$j_1$ and~$j_2$ and the total waiting time is reduced.

\begin{example}\label{example:swap}
Consider two relay robots~$j_1$ and~$j_2$ with timed paths~$(\Gamma_{j_1},T_{j_1})$ and $(\Gamma_{j_2},T_{j_2})$ as shown in Figure~\ref{fig:swap}. 
The reference velocities are given in Section~\ref{sec:simulate}. 
The paths are updated by the above algorithm to swap their meeting events.
The total waiting time is reduced from~$43.3s$ to~$12.1s$. \hfill $\blacksquare$
\end{example}

 \subsection{Real-time Execution}\label{sec:execution}
Real-time execution of  the system consists of two essential components: (i) the local plan execution of source robots and (ii) the meeting events between source and relay robots.

 \subsubsection{Plan Execution}\label{sec:plan-execute}
After the system starts, each source robot~$i\in \mathcal{N}^f$ executes its discrete plan~$\tau_{i,\mathcal{R}} = \pi_{i,\mathcal{R}}^0\,\pi_{i,\mathcal{R}}^1 \cdots \pi_{i,\mathcal{R}}^{k_i-1}\, (\pi_{i,\mathcal{R}}^{k_i}\,\pi_{i,\mathcal{R}}^{k_i+1}\cdots \pi_{i,\mathcal{R}}^{K_i})^\omega$, where~$\pi^k_{i,\mathcal{R}}=\langle \pi_{i,s_k}, g_{i,\ell_k}\rangle \in \Pi_{i,\mathcal{R}}$, $\forall k=0,1,\cdots,K_i$, which was derived in Section~\ref{sec:synthesis}.
Starting from the initial position~$\pi_{i,s_0}$,  robot~$i$ first navigates to region~$\pi_{i,s_1}$ through the corresponding path~$\Gamma_{i,s_0s_1}$. The  control inputs follow the turn-and-forward switching control: ({C.1}): $v_i = 0$ and~$\omega_i =\omega_i^{\text{ref}}$; and ({C.2}): $v_i = v_i^{\text{ref}}$ and~$\omega_i = 0$. The controller ({C.1}) is activated to turn robot~$i$ towards the next waypoint in~$\Gamma_{i,s_0s_1}$ and then, ({C.2}) drives it forward with the reference speed.

Once robot~$i$ reaches~$\pi_{i,s_1}$, it performs the data-gathering action~$g_{i,\ell_1}$ there. 
After the action is completed, robot~$i$ navigates to region~$\pi_{i,s_2}$ through~$\Gamma_{i,s_1s_2}$ and performs action~$g_{i,\ell_2}$ there.
This procedure repeats itself until robot~$i$ reaches the~($k_e$)$_{\text{th}}$ state~$\pi_{i,\mathcal{R}}^{k_e}$ according to~\eqref{eq:ke}.
During this period of time, the amount of data units stored in robot~$i$'s buffer is increased incrementally by~$D_{i}(g_{i,\ell_k})$ using~\eqref{eq:gather}, $\forall k=0,1,\cdots,k_e$.
Then on its way from state~$\pi_{i,\mathcal{R}}^{k_e}$ to~$\pi_{i,\mathcal{R}}^{k_e+1}$, robot~$i$ meets with robot~$j_i^\star$ at waypoint~$m_{j_i^\star i}$.
It is ensured by the formulation of~\eqref{eq:ke} that the buffer is never overflowed and all data-gathering actions can be performed before reaching~$\pi_{i,\mathcal{R}}^{k_e+1}$.
After the meeting, robot~$i$ continues executing the rest of its plan until the next meeting event with~$j_i^\star$ or another relay robot.
Similarly, any relay robot~$j\in \mathcal{N}^l$ starts by executing the path~$\Gamma_j$ derived from~\eqref{eq:ilp} at time~$0$, which is then modified by adding new segments each time robot~$j$ coordinates with a source robot about the next meeting event.  

\begin{remark}\label{remark:dynamics}
{Note that if a different controller is used, such as PID-based line following, then \eqref{eq:approx-time} needs to be updated to reflect the estimation of traveling time between waypoints. Furthermore, more complex robot dynamics can also be incorporated as long as the robot's traveling time between two waypoints can be well estimated.  \hfill $\blacksquare$ }
\end{remark}

\subsubsection{Meeting Events Execution}\label{sec:meet-execute}
Assume that~$\Gamma_{i,s_if_i}=m_{i,s_i}m_{i,s_i+1}\cdots m_{i,f_i}$ is the path that robot~$i$ follows to navigate from~$m_{i,s_i}$ to~$m_{i,f_i}$, and assume also that its  confirmed meeting waypoint with robot~$j_i^\star$ is $m_{i, s^\star}$.
Starting from~$m_{i,s_i}$, robot~$i$ moves towards~$m_{i,s^\star}$. 
Then \emph{two} cases are possible: (i) if robot~$j_i^{\star}$ is already waiting at~$m_{i,s^\star}$, then robot~$i$ continues moving towards~$m_{i,s^\star}$ until robot~$j_i^\star$ is within its communication range. When this happens, robot~$i$ transfers {all} the data stored in its buffer to robot~$j_i^{\star}$. 
As a result, the stored data units in the buffers of robots~$i$ and~$j_i^\star$ are updated according to
$b_i(t^+) = 0$ and $b_{j^\star}(t^+) = b_{j^\star}(t^-) + b_i(t^-)$.
  When the data transfer is completed, robot~$j_i^{\star}$ uploads all the data in its buffer to the data station immediately. Thus its stored data is updated according to
$b_{j^\star}(t^+) = 0$.
{If the stored data  at robot~$i$ is more than robot~$j_i^\star$'s buffer size~$\overline{B}_{j_i^\star}$, these data are divided into smaller batches, which are then transferred to robot~$j_i^\star$ sequentially};
(ii) if robot~$j_i^{\star}$ has not arrived at~$m_{j_i^{\star}i}$ yet, then robot~$i$ waits until robot~$j_i^{\star}$ enters its communication range and then follows the same procedure as in~(i).

{Note that due to the waiting procedure described above, an \emph{exact} synchronization on the meeting times is not required between the source and relay robots. Namely, if either robot~$i$ or $j_i^\star$  arrives at a meeting location later than the agreed meeting time~$t_{ji}$ (e.g., due to uncertainty in robot velocity), the other robot that arrives early will wait until the data exchange happens. Therefore, the proposed method can handle uncertainty in the traveling times defined in~\eqref{eq:approx-time}. 
Furthermore, delays on current meeting events do not propagate to the future meeting events since all subsequent meeting events defined in~\eqref{eq:next} are always coordinated using the current meeting times. In other words, delays are always reset to zero whenever two robots meet.  
A numerical robustness analysis of the proposed approach can be found in Section~\ref{sec:simulate}.
}

\begin{proposition}\label{theo:correctness}
Under Assumption~\ref{assum:initial} stating that each source robot is connected to \emph{at least} one relay robot initially, the above framework ensures that each source robot~$i\in \mathcal{N}^f$ can satisfy its local task~$\varphi_i$ and also that its buffer will not overflow.
\end{proposition}
\begin{proof} 
  First, the correctness of the local plan for each source robot is guaranteed by the model-checking algorithm, see~\cite{baier2008principles,guo2015multi}.
Moreover, since all local tasks are independent, these local plans can be executed independently. Thus we need to show that the plan can be executed successfully by each source robot, i.e., the data-gathering actions can be performed and the data buffer never overflows.
 Initially, each source robot is confirmed to meet with one relay robot by~\eqref{eq:ilp}.
 When the two robots meet, the stored data can be transferred and uploaded, before the source robot's buffer overflows due to the formulation of~\eqref{eq:ke}. 
Then execution of the meeting events above ensures {that every source robot always waits to meet a relay robot and transfer the stored data before performing the next gathering action that leads to buffer overflow}.
Similarly, the spontaneous meeting events described in Section~\ref{sec:spontaneous} ensure that all data-gathering actions up to the next meeting time can be performed and the data buffer never overflows. The same procedure repeats itself and holds for all source robots. 
\end{proof}

\section{Data Center Constraints, Robot Failures, and Dynamic Robot Membership}\label{sec:discussion}
In this section, we discuss how the proposed framework can be extended to account for a fixed data center, robot failure, and dynamic membership. {The later two characteristics enhance the robustness of the proposed approach.}

\subsection{Fixed Location of Data Center}\label{sec:fixed}
As mentioned in Remark~\ref{remark:immediate}, assume that a relay robot~$j$, instead of uploading its stored data immediately after meeting a source robot, needs to visit a \emph{fixed} data center~$H_j \in M$ within the workspace to upload the data, $\forall j\in \mathcal{N}^l$.
Then the proposed scheme can be modified as follows. Consider the meeting between robot~$j$ and the source robot~$i\in \mathcal{N}^f$. 
First, during the execution of the meeting event as discussed in Section~\ref{sec:meet-execute}, robot~$j$'s motion plan needs to be modified to include visiting the data center. 
In particular, if the amount of data robot~$i$ needs to transfer is less than robot~$j$'s buffer size, robot~$j$ can receive \emph{all} the data at once and then travel to the data center via the shortest path to upload the data. On the other hand, if the amount of data robot~$i$ needs to transfer is more than robot~$j$'s buffer size, robot~$j$ can receive the data in batches that equal to its buffer size, and then travel to the data center multiple times. {Consequently, for fixed data center locations, it may be beneficial to pair up source and relay robots of similar buffer sizes to reduce the number of times that relay robots need to travel to the data center. In this case, Algorithm~\ref{eq:ilp} can be modified by redefining $c_{{k_i}{k_h}}$ as follows:}
$$c_{{k_i}{k_h}}\triangleq |t_{i,{k_i}} + T_j(m_{i,{k_i}}, m_{h,{k_h}}) + N_{j i}\cdot T_{j}(m_{i, s_{j i}},H_{j})-t_{h,{k_h}}|,$$
{where~$N_{j i}\triangleq 2\cdot \left \lceil \overline{B}_{j}/\overline{B}_{i} \right \rceil$ is the number of times robot~$j$ needs to travel to the data center~$H_{j}$,~$\overline{B}_{j}/\overline{B}_{i}$ is the ratio between robot~$j$ and robot~$i$'s buffer size, the  function~$\left \lceil \cdot \right \rceil$ returns the previous largest integer; and $T_{j}(m_{i, s_{j i}},H_{j})$ is the time it takes for robot~$j$ to navigate from waypoint~$m_{i, s_{j i}}$ to $H_{j}$. As a result, the coordination obtained by the solution of problem~\eqref{eq:ilp} now considers the extra time that is needed for robot~$j$ to travel to~$H_{j}$ to empty robot~$i$'s buffer given robot~$j$'s buffer limit.}

Second, regarding the coordination of the next meeting event as discussed in Section~\ref{sec:next}, robot~$j$'s choice of the next meeting location from~\eqref{eq:next} can be modified as follows:
\begin{equation}\label{eq:next-new}
  \begin{split}
    s_{j i}^+&=\textbf{argmin}_{s_i\leq s_{j i} \leq f_i }\|t_{j,f_{j}}-t_{i, s_{j i}}\\
    &+T_{j}(m_{j,f_{j}},
\, m_{i, s_{j i}})+N_{j i}\cdot T_{j}(m_{i, s_{j i}},H_{j})\|,\\
    \end{split}
  \end{equation} 
where $N_{j i}$ and $T_{j}(H_{j},\, m_{i, s_{j i}})$ are defined above. Now~\eqref{eq:next-new} takes into account the extra time that is needed for robot~$j$ to travel to~$H_{j}$ in order to empty robot~$i$'s buffer. 
Last but not least, if there are multiple data centers that robot~$j$ can choose from, we can easily modify~\eqref{eq:next-new} to find the optimal one.

\subsection{Robot Failures}\label{sec:failure}
Let us assume first that a source robot~$i\in \mathcal{N}^f$ fails. If  robot~$i$ can still communicate with all relay robots it has committed to meet, then robot~$i$ can initiate a \emph{cancel} message to each of them to cancel the committed meeting events. In this way, these relay robots can skip the meeting with robot~$i$ and continue meeting the next source robot (instead of waiting indefinitely for robot~$i$). 
However, if robot~$i$ fails when it is not in the communication range of one or more relay robots, then  to avoid deadlock we can introduce a \emph{maximum} waiting time~$T_{\max}>0$, so that if a robot waits at a confirmed meeting location for a period of time longer than~$T_{\max}$, then it assumes that this meeting is canceled and continues executing its discrete plan until the next meeting event. 

Assume now that a relay robot~$j\in \mathcal{N}^l$ fails. If robot~$j$ can still communicate with the source robots it is committed to meet, it can cancel the meeting events directly as before.
However, in this case, the source robot~$i\in \mathcal{N}_j^f$ can \emph{not} simply skip this meeting event and continue its plan execution as its buffer will overflow. Instead, robot~$i$ needs to navigate to its next meeting location directly, upload its stored data with relay robot $j'\in \mathcal{N}^l$ and {more importantly keep the next meeting event with robot~$j'$ unchanged. In other words, robot~$i$ needs to meet with robot~$j$ consecutively twice.} Last but not least, if robot~$j$ is the \emph{only} relay robot that robot~$i$ is committed to, robot~$i$ may have to wait until it meets another relay robot to upload its data. This can only happen spontaneously as robot~$i$ has no knowledge of the location of other relay robots due to limited communication range. {This situation can be solved by allowing source robots to relay data to each other or exchange information about their meeting events, which is part of our ongoing work, see also Remark~\ref{remark:inter-type-A}.}

{At last, if several source or relay robots fail, the procedure described above will be performed for each fault robot.}
Moreover, if a robot recovers after failure, it will be treated as a robot that newly joins the system, as discussed below. 

\subsection{Dynamic Membership}\label{sec:new}

By dynamic membership, we mean that (i) existing robots within the team can leave the team without resulting in a deadlock; and (ii) new robots can join the team seamlessly without the need to restart the system. 
The first case can be achieved in a similar way as described in Section~\ref{sec:failure} to handle robot failures. Particularly, before a source robot leaves the team, it needs to meet with each relay robot that it is committed to meet, but \emph{without} coordinating the next meeting event.
In the same way, before a relay robot leaves the team, it still needs to meet with each source robot that it is committed to meet, without coordinating the next meeting event. Secondly, due to the distributed and online nature of the proposed scheme, new source or relay robots can be easily added to the system during run time. If the new relay robot~$j'$ that just joined the team is connected to an existing source robot~$i\in \mathcal{N}^f$, robot~$i$ will treat this meeting as a spontaneous meeting event as described in Section~\ref{sec:spontaneous}. The same procedure applies when an existing relay robot meets a new source robot that just joined the team during run time. {However, a new source robot must be connected to at least one relay robot when it joins the team.}

\section{Case Study}\label{sec:simulate}
This section presents simulation results for a team of 12 data-gathering robots.
All algorithms are implemented in Python 2.7. ``\texttt{Gurobi}''~\cite{gurobi} and~``\texttt{poly2tri}''~\cite{poly2tri} are external packages {and~``\texttt{P\_MAS\_TG}''~\cite{package} is developed by the authors.}
All simulations are carried out on a laptop (3.06GHz Duo CPU and 8GB of RAM).

\subsection{System Description}\label{sec:sys-description}
All 12 robots satisfy the unicycle dynamics~\eqref{eq:unicycle}. There are 9 source robots (denoted by~$a_0,   a_2, \cdots, a_{8}$) and 3 relay robots (denoted by~$l_{1}, l_2,  l_3$). 
The workspace has size~$10m\times 10m$ and contains three polygonal obstacles, as shown in Figure~\ref{fig:triangulation}.  The triangular partition is derived from~\cite{poly2tri}. 
All robots' communication ranges are set to~$1m$. The reference linear and angular velocities are chosen randomly between~$[0.5,  0.8]m/s$ and $[0.1,  0.3]rad/s$. 
The buffer size of all source robots is chosen randomly between~$[3, 5]$ data units, while all relay robots have a buffer size of 5 data units.

To simplify the task description, we divide the source robots into three categories: (i) the first category ($a_0, a_1, a_2$) gathers type-1 data in region~$r_1$, type-2 data in region~$r_2$ and type-3 data in region~$r_3$ (in any order), infinitely often. This specification can be expressed by the LTL formula~$\varphi_{c_1}= \square\Diamond (r_2\wedge g_2) \wedge \square\Diamond (r_1\wedge g_1) \wedge \square\Diamond (r_3\wedge g_3)$;
(ii) the second category ($a_3, a_4, a_5$) gathers type-4 and type-5 data in regions~$r_4$, then type-4 data in region~$r_6$ (in this order) and also type-5 data in region~$r_5$, infinitely often, i.e., $\varphi_{c_2}=\square \Diamond (((r_4 \wedge g_4) \wedge \bigcirc (r_4 \wedge g_5)) \wedge \Diamond (r_6\wedge g_4)) \wedge \square \Diamond (r_5 \wedge g_5)$;
(iii) the third category ($a_6, a_7, a_8$) gathers type-6 data in regions~$r_7$, $r_9$ and type-7 data in region~$r_8$, infinitely often, i.e.,~$\varphi_{c_3}=\square \Diamond (r_8 \wedge g_7) \wedge \square \Diamond (r_7 \wedge g_6)\wedge \square \Diamond (r_9 \wedge g_6)$.

The actions~$g_2, g_3, g_4, g_6$ gather~$2$ units of data, while actions~$g_1, g_5, g_7$ gather~$1$ unit.
Moreover, any data-gathering action takes~$1s$ while the data transfer or upload actions take~$2s$. 
Initially, robots $a_0, a_3, a_6, l_1$ start from $(6.5m,   6.6m)$, robots $a_1, a_4, a_7, l_2$ start from $(5.6m,   5.0m)$, and robots $a_2, a_5, a_8, l_3$ from $(4.6m,   4.3m)$.  
Thus every source robot is connected to at least one relay robot, as required by Assumption~\ref{assum:initial}. 

\begin{figure}[t]
\centering
   \includegraphics[width =0.35\textwidth]{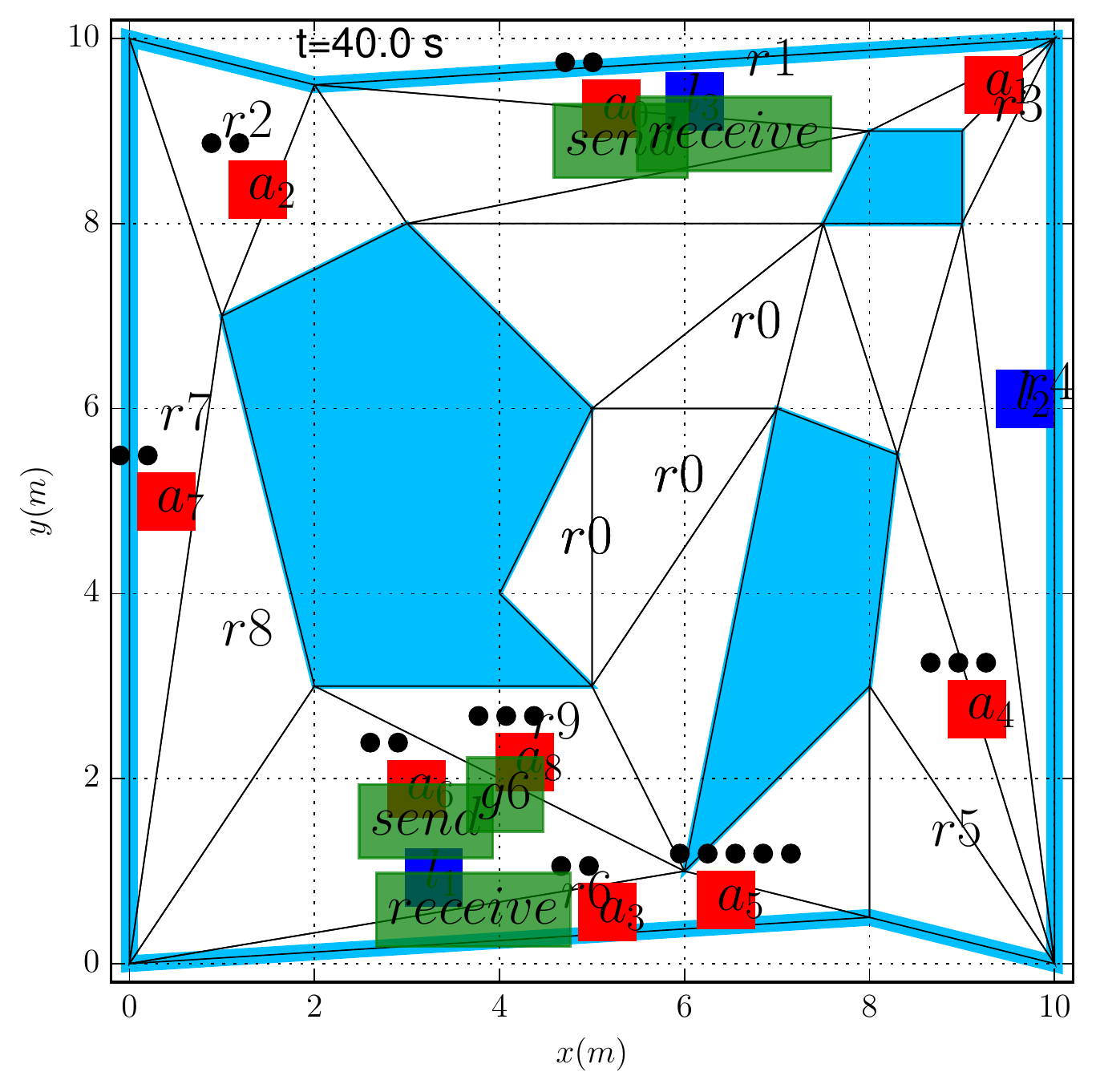}
  \caption{{A} snapshot of the simulation at~$40s$. source and relay robots are red and green squares, while the stored data units are indicated by black circles. The data-gathering actions, data transfer and upload actions are shown by filled green text boxes, e.g., ``$g_6$, send, receive, upload''. All robots and regions of interest are labeled by their names.}
\label{fig:snaps}
\end{figure}


\subsection{Simulation Results}\label{sec:sim-results}
First, the roadmap of each robot is constructed using a triangular partition of the workspace, as described in Section~\ref{sec:roadmap}.
For robots~$a_0,a_1,a_2$, the FTS~$\mathcal{R}_i$ has $16$ nodes and $112$ edges, the NBA~$\mathcal{A}_{\varphi_i}$ has~$4$ nodes and $13$ edges, and the product~$\mathcal{P}_i$ has~$64$ nodes and~$476$ edges.
For robots~$a_3,a_4,a_5$, the FTS~$\mathcal{R}_i$ has $12$ nodes and $72$ edges, the NBA~$\mathcal{A}_{\varphi_i}$ has~$7$ nodes and $32$ edges, and the product~$\mathcal{P}_i$ has~$84$ nodes and~$342$ edges.
For robots~$a_6,a_7,a_8$, the FTS~$\mathcal{R}_i$  has $12$ nodes and $72$ edges, the NBA~$\mathcal{A}_{\varphi_i}$ has~$4$ nodes and $13$ edges, and the product~$\mathcal{P}_i$ has~$48$ nodes and~$312$ edges.

\begin{figure}[t]
\centering
   \includegraphics[width =0.48\textwidth]{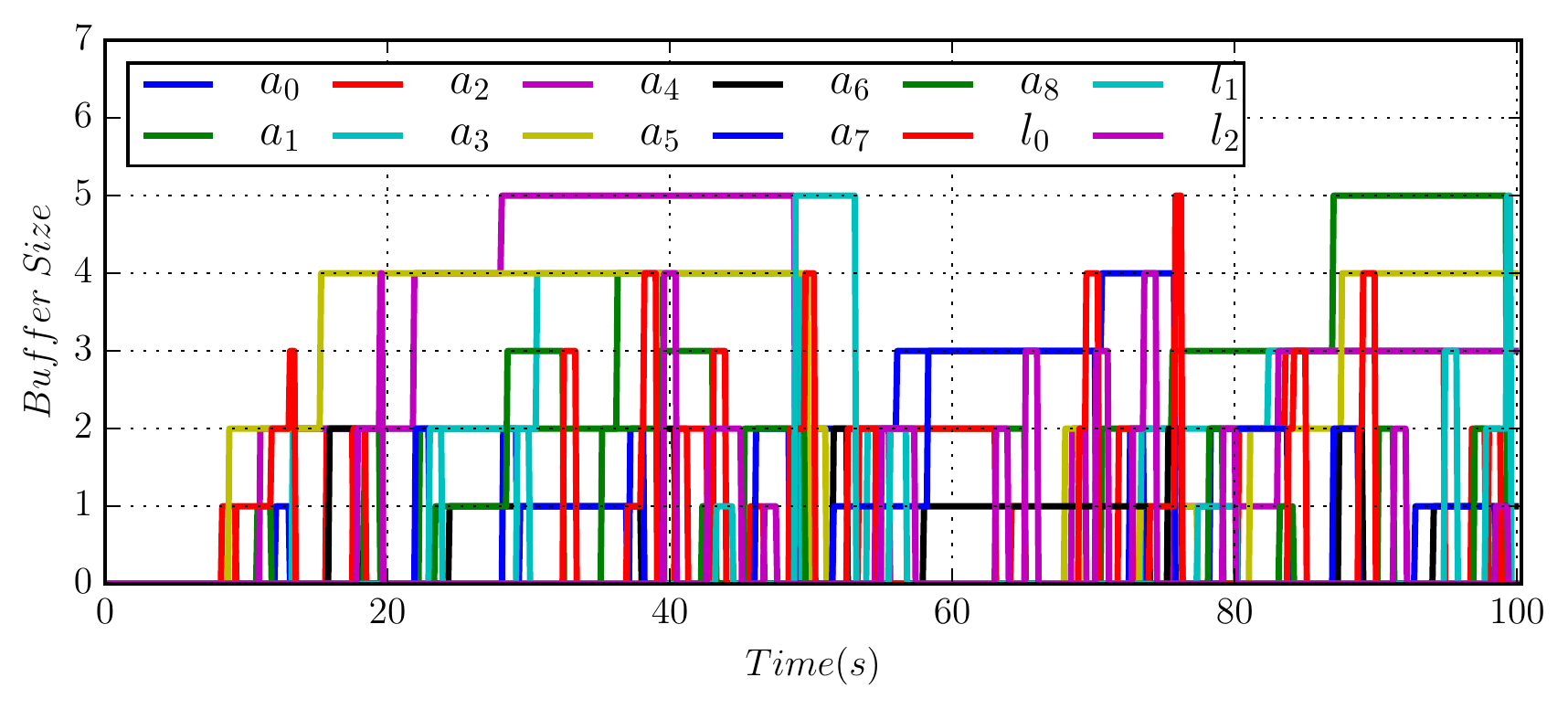}
  \caption{Stored data at each robot's  buffer during the simulation. The buffer sizes of robots~$a_0,a_1,\cdots,a_8$ and~$l_0,l_1,l_2$ are set to $[4,  5,  3,  4,  5,  5,  4,  5,  3,  5,  5,  5]$, which are respected for all time.}
\label{fig:buffer}
\end{figure}

Then each source robot synthesizes its discrete plan using the algorithm in~\cite{guo2015multi} and the package~\cite{package}. 
It took approximately~$0.03s$, $0.05s$ and~$0.01s$ for the above three groups to synthesize their discrete plans.
For instance,~$a_0$ has prefix cost~$57.22$ and suffix cost~$46.14$, while~$a_3$ has prefix cost~$60.60$ and suffix cost~$45.69$. 
It took $0.3s$ by Gurobi~\cite{gurobi} to find the optimal solution of~\eqref{eq:ilp}, which determines the initial paths of all relay robots.
The discrete plans are executed according to Section~\ref{sec:plan-execute}, while the data are transferred and uploaded during the meeting events as described in Section~\ref{sec:meet-execute}.
The coordination for the next meeting event and spontaneous meetings follow Sections~\ref{sec:next} and~\ref{sec:spontaneous}.
We simulate the system for~$100s$. {A snapshot} of the simulation at~$40s$ is shown in Figure~\ref{fig:snaps}, where we show the number of data units stored at each robot's buffer and the action taken by each robot.
The evolution of the stored data units at each robot's buffer is shown in Figure~\ref{fig:buffer}. 
The maximum number of connected robots remains below~$5$ during most of the simulation, as shown in Figure~\ref{fig:max-component}. Thus the communication network among the robots is almost \emph{never} connected.
Furthermore, we also monitor the times that  relay robots~$l_0,l_1,l_2$ swap meeting events as described in Section~\ref{sec:swap}. Figure~\ref{fig:swap-event} shows the reduction in total waiting time after two relay robots swapping their meeting events.
In total,~$137$ units of data are uploaded, as shown in Table~\ref{table:compare-data} and Figure~\ref{fig:total}.
The complete simulation videos can be found in~\cite{tro17-videos}.

\begin{figure}[t]
\centering
   \includegraphics[width =0.35\textwidth]{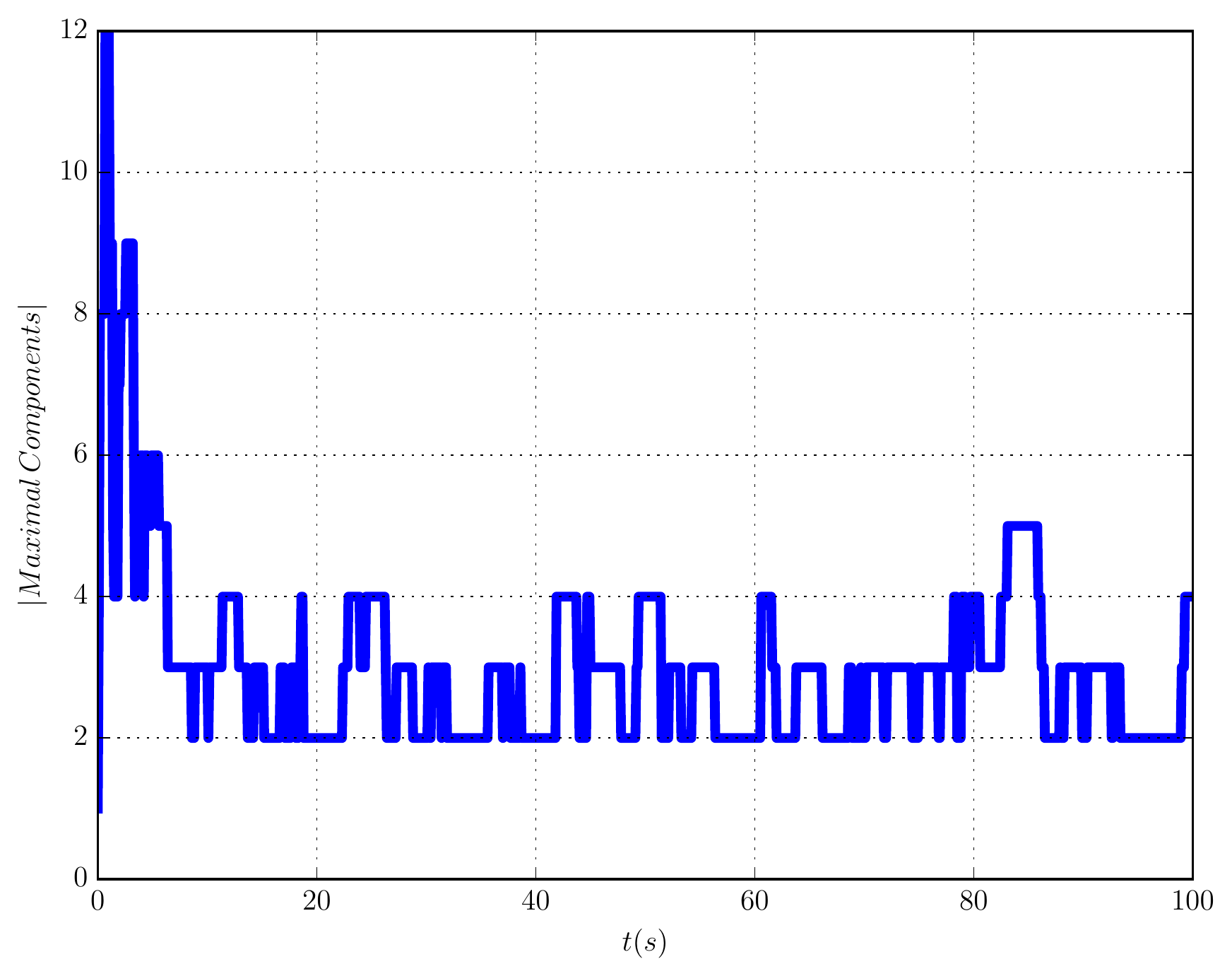}
  \caption{The evolution of the size of maximal components of the communication graph (of size~$12$), i.e., the maximal number of connected robots, given the uniform communication range~$1m$.}
\label{fig:max-component}
\end{figure}

\begin{figure}[t]
\centering
   \includegraphics[width =0.5\textwidth]{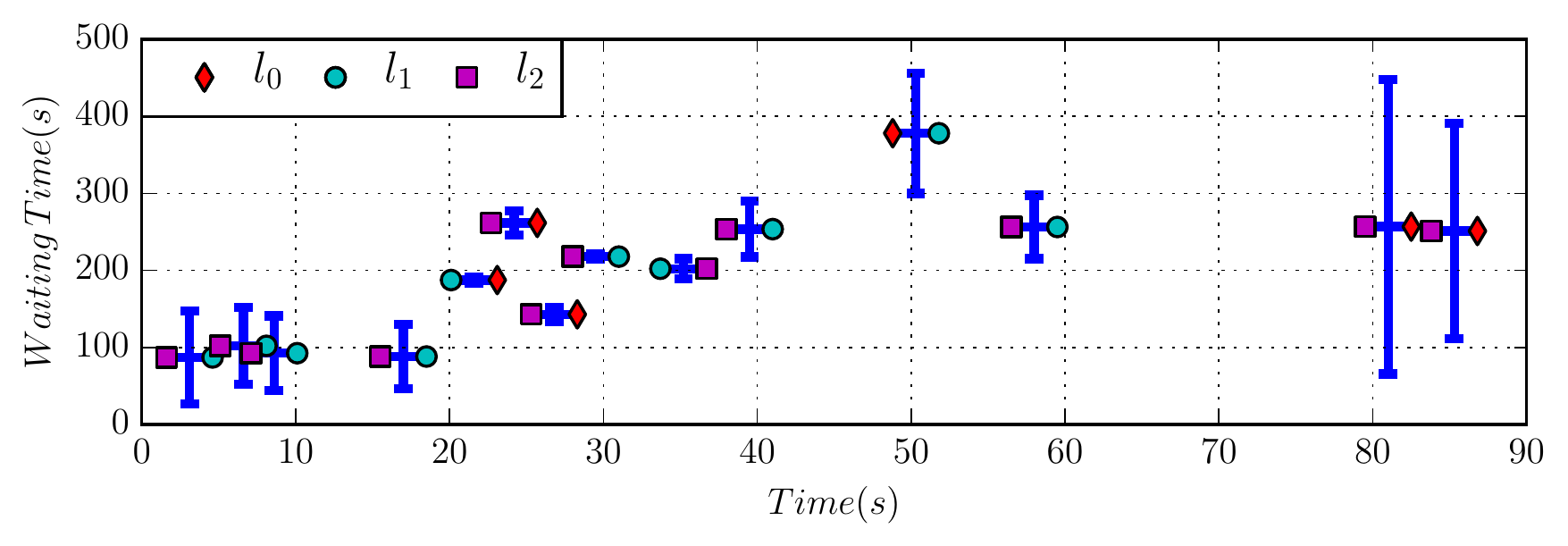}
  \caption{History of relay robots~$l_0,l_1,l_2$ swapping meeting events during the simulation. The high and low points of the error bar indicate the total waiting time before and after the swapping, respectively.}
\label{fig:swap-event}
\end{figure}

\begin{figure}[t]
\begin{minipage}[t]{0.495\linewidth}
\centering
   \includegraphics[width =1.02\textwidth]{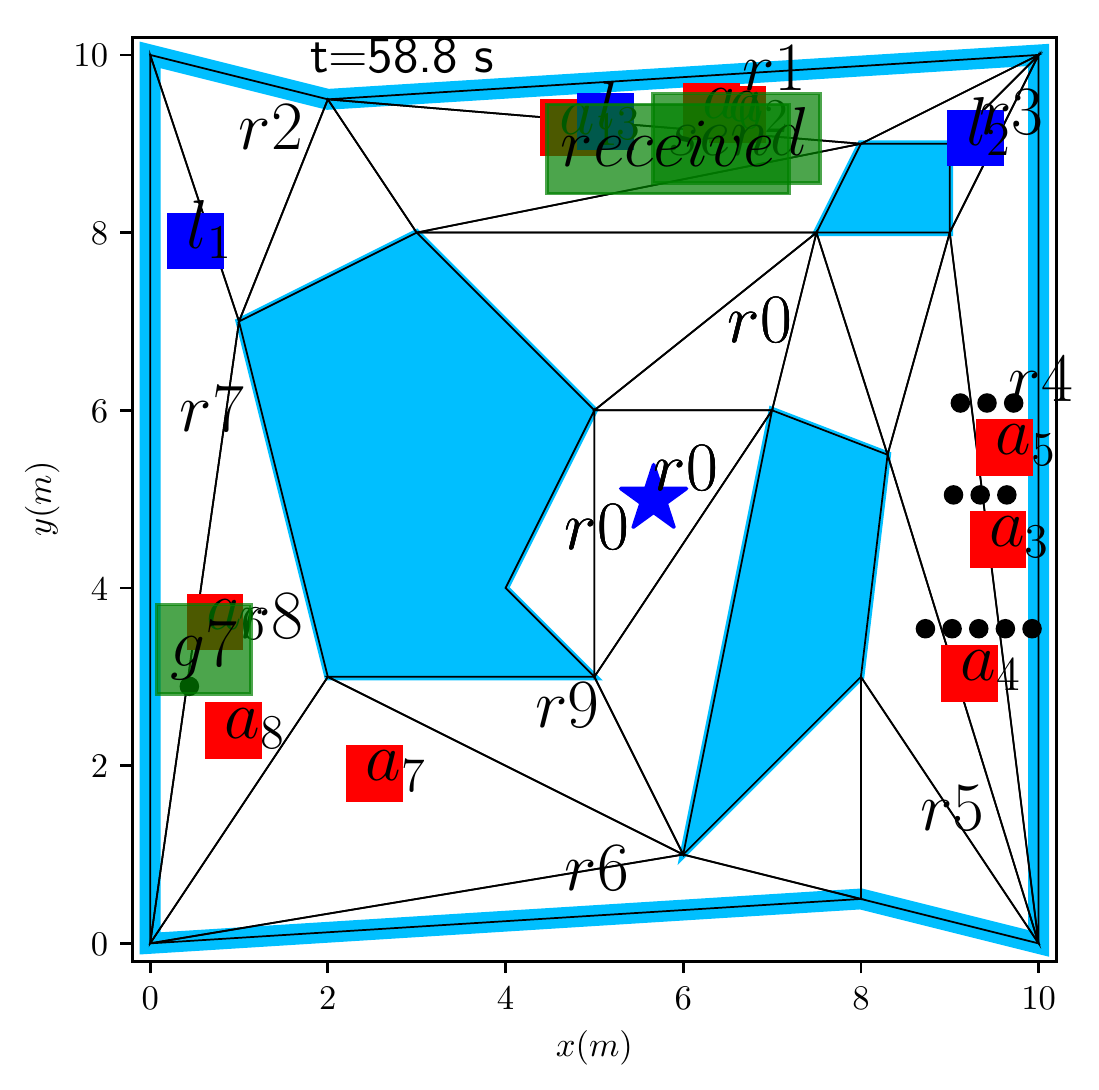}
  \end{minipage}
\begin{minipage}[t]{0.495\linewidth}
\centering
    \includegraphics[width =1.02\textwidth]{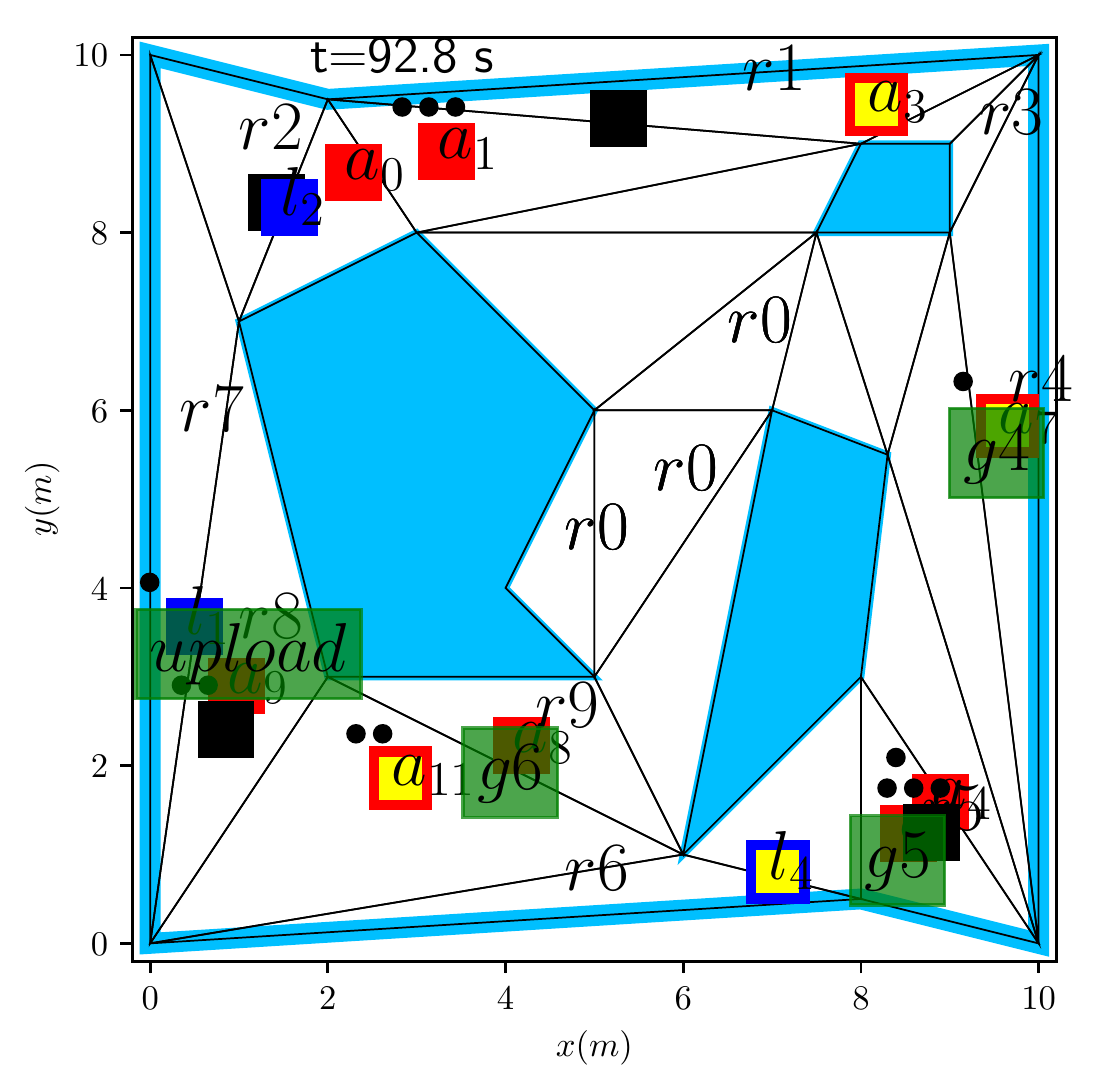}
  \end{minipage}
  \caption{Left: a snapshot of the simulation where a fixed upload center is given for each relay robot (marked by the blue star); Right: a snapshot of the simulation where existing robots fail (in black squares) and new robots join the team (marked by yellow squares).}
\label{fig:extend}
\end{figure}

\begin{figure}[t]
\centering
   \includegraphics[width =0.48\textwidth]{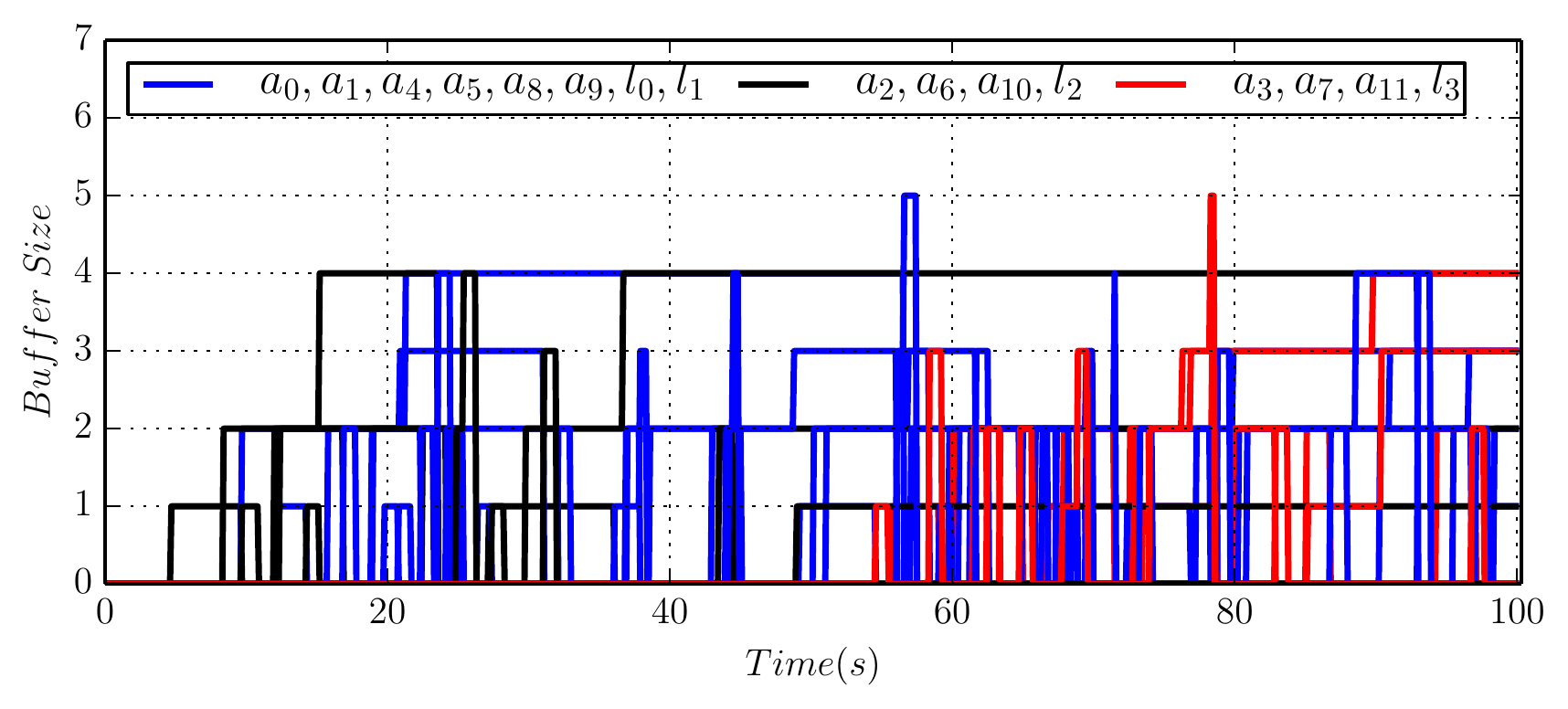}
  \caption{Stored data at each robot's  buffer for the scenario where three source robots  $a_2,a_6,a_{10},l_2$ fail at time~$t=50s$ (in black lines). At the same time robots $a_3,a_7,a_{11},l_3$ join the team (in red lines). The other robots are shown in blue lines. The buffer sizes are set to $[4,  5,  3,  5, 4,  5,  5,  5, 4,  5,  3,  5, 5,  5,  5, 5]$, which are all respected.}
\label{fig:buffer_new_member}
\end{figure}
{Furthermore, in order to demonstrate the robustness of the proposed approach to uncertainties in the robots' traveling times, we have simulated the case where the traveling velocity of all robots is subject to additive random noise (with zero mean and variance equal to $20\%$ of the velocity value.). As shown in Figure~\ref{fig:total}, the delays in the meeting events caused by uncertain traveling times are not propagated across the network and the total amount of gathered data within~$100s$ in this case is $96$ (close to $137$ in the nominal case).}

Last but not least, as discussed in Section~\ref{sec:discussion} and shown in Figure~\ref{fig:extend}, the proposed scheme can be easily extended to take into account other scenarios, e.g., fixed data center, robot failures and new members.
First, we choose a fixed data center located at coordinate $(8.3,\,7.2)$ that all relay robots need to visit to upload its stored data. Instead of uploading the data directly, a local planning module is used by each relay robot to navigate to this fixed data center as proposed in Section~\ref{sec:fixed}.
{Second, we introduce faults to source robots~$a_2$, $a_6$, $a_{10}$ and relay robot~$l_2$ at time~$50s$, when they all stop moving and remain static. 
Moreover, three new source robots~$a_3,a_7,a_{11}$ and one relay robot~$l_3$ are added to the system (thus 16 robots in total), with the source robots having the same task description as three groups described earlier. }
The evolution of the stored data at each robot's buffer is shown in Figure~\ref{fig:buffer_new_member}. It shows that the buffer of these faulty robots remains unchanged after the faults occur, while the rest of the team (along with the new members) follow the reconfiguration scheme from Sections~\ref{sec:failure} and~\ref{sec:new} while respecting the buffer constraint. It can be seen from the simulation results that for the clustered workspace considered here, the meeting events with a faulty robot are canceled once the maximum waiting time is reached and furthermore the new robots can easily join the network via the spontaneous meeting events.
Simulation videos under these extended scenarios can be found in~\cite{tro17-videos}.

\subsection{Comparisons to Other Approaches}\label{sec:compare}

In this part, we compare the data-gathering performance of the proposed scheme to the centralized approach and two static approaches introduced below. Simulation videos for all three approaches can be found in~\cite{tro17-videos}.

\subsubsection{Centralized Approach}\label{sec:centralized-comp}
{As mentioned in Section~\ref{sec:centralized}, the centralized solution provides the optimal solution in terms of total distance traveled.} For this case study, the product motion model has approximately~$16^3\cdot 12^3\cdot 12^3 \cdot 24^3\approx 1.6\times 10^{14}$ states and $112^3\cdot 72^3\cdot 72^3 \cdot 70^3\approx 7.3\times 10^{22}$ transitions. The product B\"uchi automaton has approximately $4^3\cdot 7^3\cdot 4^3 \approx 1.4\times 10^6$ states and $13^3\cdot 32^3\cdot 13^3 \approx 1.5\times 10^{11}$ transitions. Thus to construct the product automaton for the whole system is computationally infeasible. 
{Moreover, we provide a numerical analysis to compare the optimality and computational complexity of the proposed approach to the centralized method, for problems of smaller size that can be handled using the centralized method. The results are shown in Table~\ref{table:compare-opt}. It can be seen that (i) for small systems (with $3-5$ robots) the centralized method provides an optimal solution that has a slightly smaller total cost of the plan suffix than the proposed approach. However, as mentioned in Section~\ref{sec:centralized}, this centralized plan can only be executed in a synchronized way and is not robust to robot failures; (ii) for larger systems (with more than 3 robots), the centralized method fails to provide a solution within reasonable time (where $\mathcal{P}$ has more than $1$ billion states), while in contrast our approach can scale much better, even to system with $7$ relay robots and $21$ source robots.}

\begin{table}[t]
\begin{center}

{\renewcommand{\arraystretch}{1.4}
\scalebox{1.05}{
    \begin{tabular}{| c| c | c |c |c|}
    \hline
   \textbf{Method} & $(\mathcal{N}^l, \mathcal{N}^f)$ & $\mathcal{P}$ & $C_{\texttt{suf}}$ & Time[s] \\ 
        \hhline{|=|=|=|=|=|}
 \multirow{4}{*}{Proposed} & $(1, 1)$ & $(64,4.7\texttt{e}3)$ & 37.2 & $0.1$s \\
\cline{2-5}
  & $(1, 2)$ & $(128,9.4\texttt{e}3)$ & 39.5 & $0.1$s \\
\cline{2-5}
  & $(1, 3)$ & $(3.6\texttt{e}3,1.3\texttt{e}4)$ & 46.7 & $0.18$s \\
\cline{2-5}
  & $(5, 15)$ &$(1.8\texttt{e}4,6.7\texttt{e}4)$  & 59.1 & $3.5$s \\
\cline{2-5}
  & $(7, 21)$ & $(2.7\texttt{e}4,9.2\texttt{e}4)$ & 67.9 & $575$s \\
        \hhline{|=|=|=|=|=|}
 \multirow{3}{*}{Centralized} & $(1, 1)$ & $(2.5\texttt{e}3,4.4\texttt{e}4)$ & 34.6 & $13.5$s \\
\cline{2-5}
 & $(1, 2)$ & $(3.1\texttt{e}5,2.2\texttt{e}7)$ & 36.4 & $16.5$h \\
\cline{2-5}
 & $(1, 3)$ & $>(5.2\texttt{e}6,1.3\texttt{e}9)$ & * & $>20$h \\
 \hline
    \end{tabular}
}}
\caption{{A comparison of optimality and computational complexity  between the proposed method and the centralized approach. The notation~${a}\texttt{e}{b}\triangleq a\times 10^{{b}}$ for~${a},{b}>0$. For the proposed method, $\mathcal{P}$ is the summation of all local product $\mathcal{P}_i$ between $\mathcal{R}_i$ and $\mathcal{A}_{\varphi_i}$, $C_{\texttt{suf}}$ is the maximum length of the plan suffix among all robots, and the synthesis time is mainly the time needed to solve the MILP problem for initial coordination. For the centralized case, $\mathcal{P}$ is the product of $\mathcal{T}_{\texttt{a}}$ and $\mathcal{A}_{\varphi_{\texttt{a}}}$ from Section~\ref{sec:centralized}, $C_{\texttt{suf}}$ is the minimum length of its plan suffix, and the synthesis time is mainly the time needed for the model-checking process.}}
\label{table:compare-opt}
\end{center}
\end{table}

\subsubsection{Static Approaches}\label{sec:trivial}

Alternatively, a straightforward solution to the data-gathering problem considered in this paper is to require that all relay robots remain static at their initial positions for all time. As a result, as long as each source robot is informed about the location of at least one relay robot, every source robot can simply navigate to the closest relay robot once it has gathered enough data that needs to be transferred and uploaded. This static approach is \emph{always} {feasible} for the problem considered here, but can be very inefficient if the workspace is large and many relay robots are located close to each other.
{The optimal placement of relay robots can only be determined in a centralized way as described in Section~\ref{sec:centralized}.} 
We implement the above approach and simulate the system for~$100s$ under the same settings presented in Section~\ref{sec:sys-description}. As a result, $58$ units of data are uploaded in total, as shown in Table~\ref{table:compare-data} and Figure~\ref{fig:total}, compared with $137$ units via the proposed dynamic approach. 
The difference is that in our approach every relay robot can actively navigate to meet multiple source robots that need to transfer data while minimizing the total waiting time.

Finally, another simple solution  is to force all source and relay robots  to move as a group that is within communication range for all time. In this case the source robots can follow a predefined static order to execute their local plans. Since all relay robots are within the communication range, the data gathered by any source robot can be transferred to any relay robot and uploaded directly. This static approach  imposes all-time connectivity of the communication network.   It can also be very inefficient since the source robots can not execute their local plans simultaneously and independently, while relay robots are not fully utilized regarding their data-uploading ability.
{This predefined static order can be also optimized in a centralized way, as described in Section~\ref{sec:centralized}, by adding the constraints that all robots are within each other's communication range.}
We implement the above approach and simulate the system for~$100s$ under the same settings. The source robots take turns to execute their local plans according to the order of their IDs.
As shown in Table~\ref{table:compare-data} and Figure~\ref{fig:total}, only $8$ units of data are uploaded in total, compared to $137$ units via our approach. The difference is that the proposed intermittent communication framework allows all source robots to move and execute their local plans independently. Thus the source and relay robots only meet when they need to transfer data and coordinate their next meeting event.

The above  studies show that the proposed dynamic approach has a much less computational burden compared to the centralized approach and improves greatly the overall data-gathering efficiency compared to the static approaches.

\begin{table}[t]
\begin{center}

{\renewcommand{\arraystretch}{1.4}
\scalebox{1.05}{
    \begin{tabular}{| c| c | c | c |c |}
    \hline
   \textbf{Approach} & type-1,2,3  &  type-4,5 & type-6,7 & \textbf{Total} \\ 
        \hhline{|=|=|=|=|=|}
 Proposed & 54 & 38 & 45 & 137 \\
\hline
 Static One & 21 & 18 & 19 & 58 \\
\hline
 Static Two & 2 & 4 & 2 & 8 \\
    \hline
    \end{tabular}
}
}
\caption{{Total amount of different types of data uploaded by the relay robots during the simulation of~$100s$, under the proposed approach and two static approaches discussed in Section~\ref{sec:compare}.}}

\label{table:compare-data}
\end{center}
\end{table}
\begin{figure}[t]
\centering
   \includegraphics[width =0.48\textwidth]{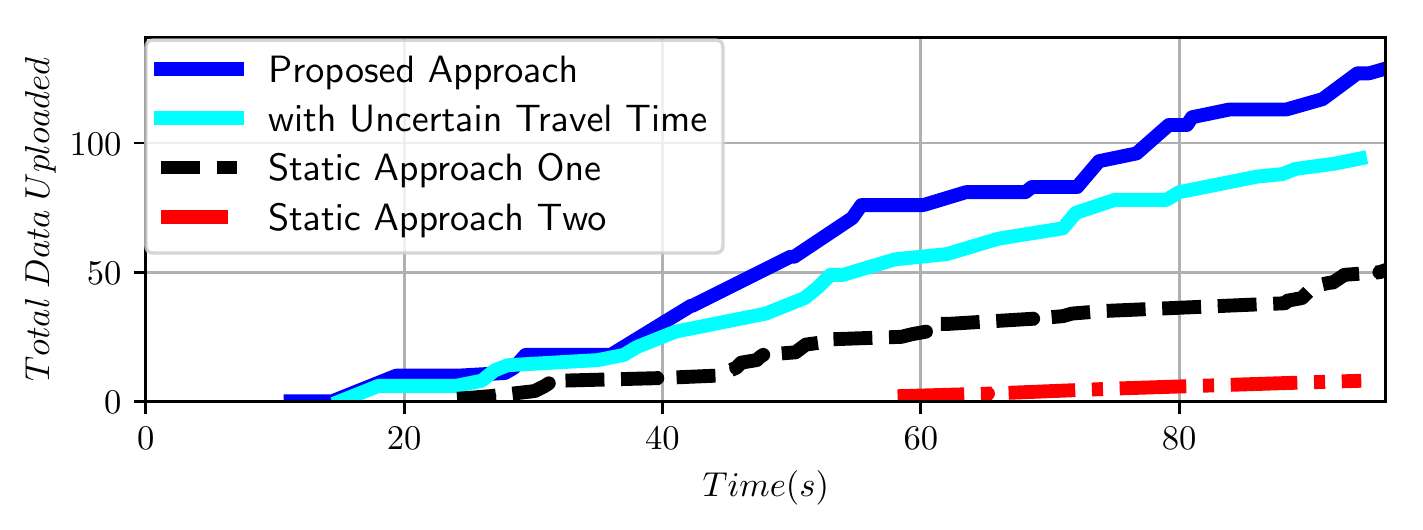}
  \caption{{The total amount of data uploaded under the proposed approach and two static approaches discussed in Section~\ref{sec:compare}. Simulation videos for all three cases are online~\cite{tro17-videos}.}}
\label{fig:total}
\end{figure}

\section{Experimental Study}\label{sec:exp}
In this section, we present the experimental study to validate the proposed approach. Four differential-driven ``iRobots'' are deployed within a~$2.5m \times 2.0 m$ workspace, as shown in Figure~\ref{fig:dynamic-ws-exp}, whose positions and orientation are tracked via an Optitrack motion capture system. 
The communication among the robots is handled by the Robot operating system (ROS).

\begin{figure}[t]
     \centering
     \includegraphics[width=0.48\textwidth]{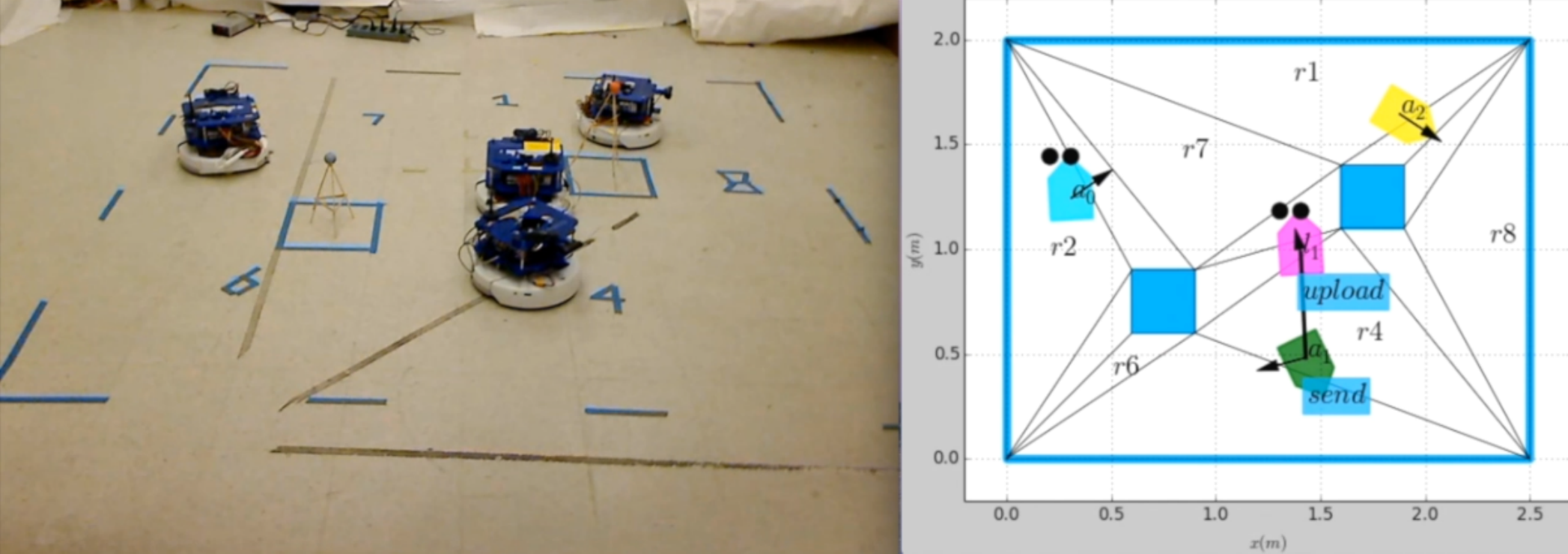}
     \caption{{A snapshot} of the experiment setup. Left: regions of interest are marked by their IDs on the ground. Tripods in boxed area are obstacles. The relay robot is marked by a yellow tape and the rest are source robots. 
Right: the real-time visualization panel to monitor the robot motion and communication. Robots~$a_0, a_1, a_2,l_1$ are in blue, green, yellow and magenta, respectively. The stored data units are indicated by filled black circles. The data-gathering, data-transfer and upload actions are shown by blue text boxes.}
     \label{fig:dynamic-ws-exp}
   \end{figure}
\begin{figure}[t]
     \centering
     \includegraphics[width=0.4\textwidth]{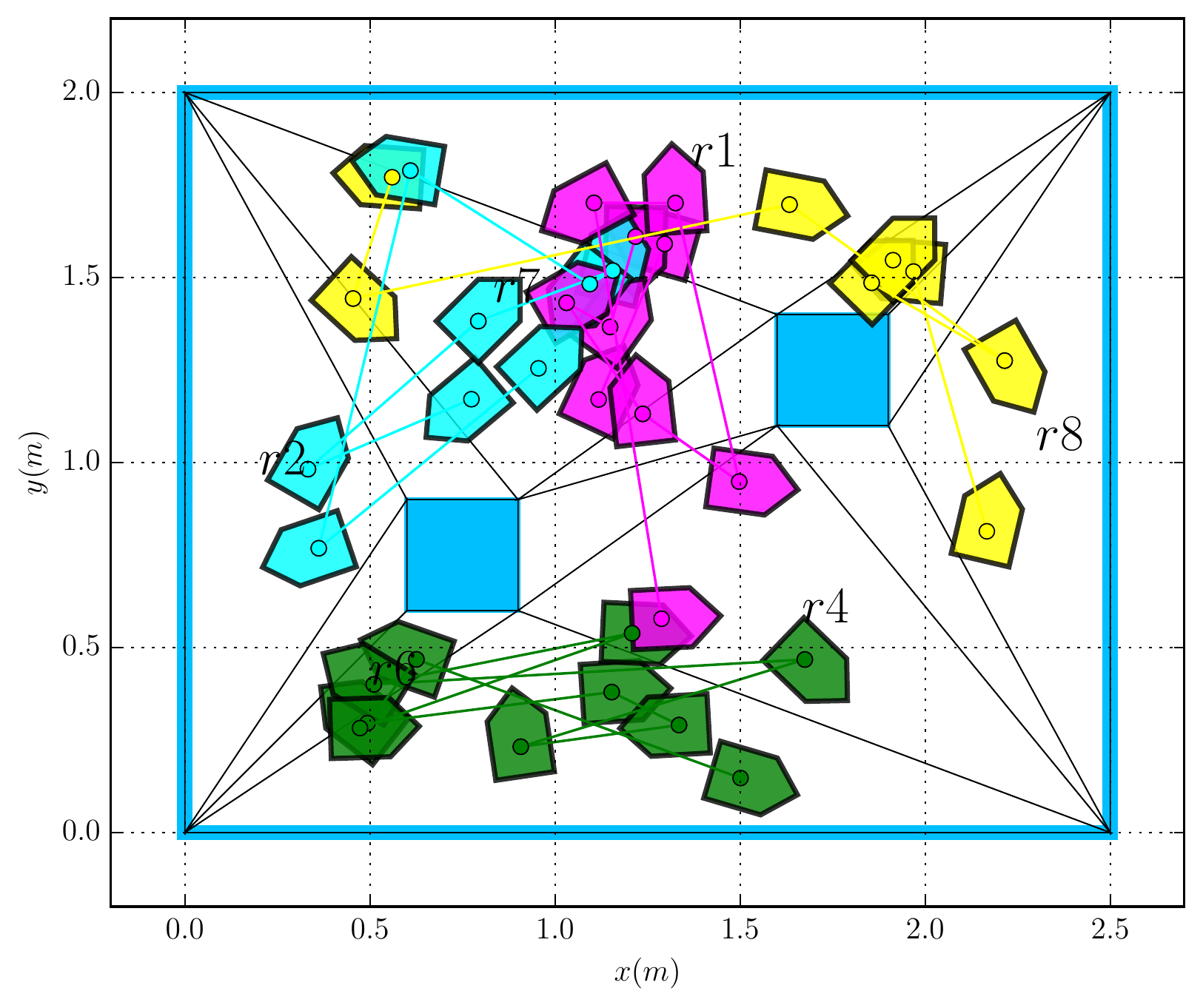}
     \caption{The trajectory of each robot during the experiment, sampled at every~$30s$. Robots~$a_0, a_1, a_2,l_1$'s trajectories are shown in blue, green, yellow and magenta, respectively}
     \label{fig:dynamic-traj-exp}
   \end{figure}


\subsection{System Description}\label{sec:model-exp}

Three iRobots serve as source robots (denoted by~$a_0, a_1, a_2$) while one serves as the relay robot (denoted by~$l_1$). 
As shown in Figure~\ref{fig:dynamic-ws-exp}, there are six regions of interest and two obstacles within the workspace; and a visualization panel is used to monitor the robot data-gathering actions and communications in real time. 
For source robots, their regions of interest, allowed actions and local tasks are defined as follows:
Robot~$a_0$ has two regions of interest~$r_1,r_2$ and two actions~$g_1,g_2$ associated with one type-1 and two type-2 data units, respectively. Its task is to gather type-1 data in region~$r_1$ and then type-2 data in region~$r_2$ (in this order) infinitely often, i.e.,~$\varphi_0=\square \Diamond ((r_1 \wedge g_1) \wedge \Diamond (r_2 \wedge g_2))$;
Robot~$a_1$ has two regions of interest~$r_4,r_6$ and two actions~$g_3,g_4$ associated with two type-3 and one type-2 data units, respectively. Its task is to gather type-3 data in region~$r_4$ and then type-4 data in region~$r_6$ (any order) infinitely often, i.e.,~$\varphi_1=\square \Diamond (r_4 \wedge g_3) \wedge \square \Diamond (r_6 \wedge g_4)$;
Robot~$a_2$ has two regions of interest~$r_7,r_8$ and two actions~$g_5,g_6$ associated with two type-5 and one type-6 data units, respectively. Its task is to gather type-5 data in region~$r_7$ and then type-6 data in region~$r_8$ (any order) infinitely often, i.e.,~$\varphi_2=\square \Diamond (r_7 \wedge g_5) \wedge \square \Diamond (r_8 \wedge g_6)$.
All robots have a limited buffer size of~$4$ data units and a communication range of~$0.8m$.
The initial position of robots~$a_0,a_1,a_2,l_1$ is given by~$(1.1,0.8),(1.1,0.2),(2.0,0.7),(1.6,0.5)$ in meters, respectively. Thus the relay robot~$l_1$ is initially connected to all source robots~$a_0,a_1,a_2$, which satisfies Assumption~\ref{assum:initial}.

The size of an iRobot is around~$0.4m$ in diameter. Given the cluttered workspace, a local collision avoidance scheme is needed for successful point-to-point navigation as an important part of the plan execution. 
In this work, we rely on the method of reciprocal velocity obstacles (RVO) introduced in~\cite{van2008reciprocal}. However, since the original algorithm is developed mainly for nonholonomic robots, not for the unicycle robots considered here, we need to introduce a transition period during which the robots turn in place towards the desired direction determined by the RVO method, before moving forward. 

\subsection{Experiment Results}\label{sec:results-exp}

Following the procedure described in Section~\ref{sec:execution}, we first synthesize the offline plan for each source robot. For robot~$a_0$, it took~$0.01s$ for the solver~\cite{package} to obtain the initial plan; similarly for~$a_1,a_2$. 
For the initial coordination via~\eqref{eq:ilp}, it took~$0.16s$ for Gurobi~\cite{gurobi} to find the initial path for robot~$l_1$. 
Once the robots starts moving, the plan execution and coordination of meeting events during run time follows Section~\ref{sec:execution}. Note that swapping meeting events between relay robots is not considered as there is only one relay robot.
The experiment was performed for a duration of~$3$ minutes, and the full  video can be found online at~\cite{tro17-videos}. 
The sampled trajectory of each robot is plotted in Figure~\ref{fig:dynamic-traj-exp}. It can be seen that each robot satisfies its local task and avoids collisions with the static obstacles. 
Moreover, the amount of data stored within each robot's buffer is shown in Figure~\ref{fig:buffer-exp}, which verifies that buffer constraints are always respected. Finally, during the experiment, $27$ data units were uploaded in total to the data center, as shown in Figure~\ref{fig:total-exp}.
\begin{figure}[t]
     \centering
     \includegraphics[width=0.48\textwidth]{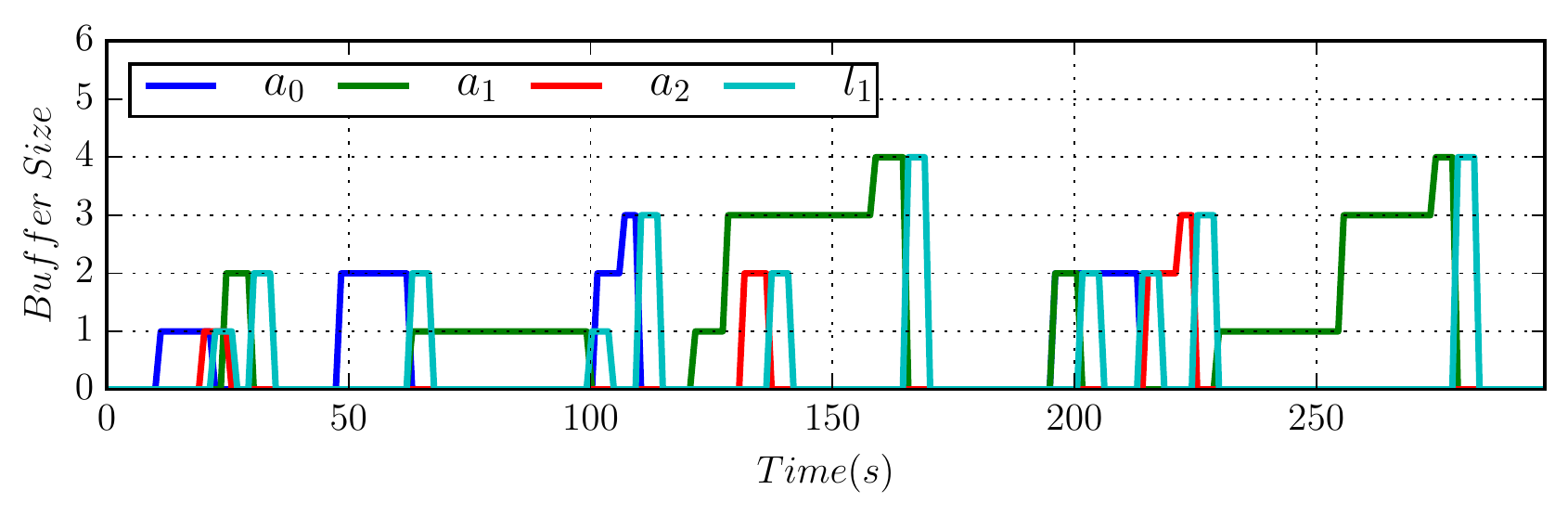}
     \caption{Evolution of the amount of data stored at each robot's buffer. Note that the buffer size limit is set to~$4$ for all robots. }
     \label{fig:buffer-exp}
   \end{figure}

\begin{figure}[t]
     \centering     \includegraphics[width=0.48\textwidth]{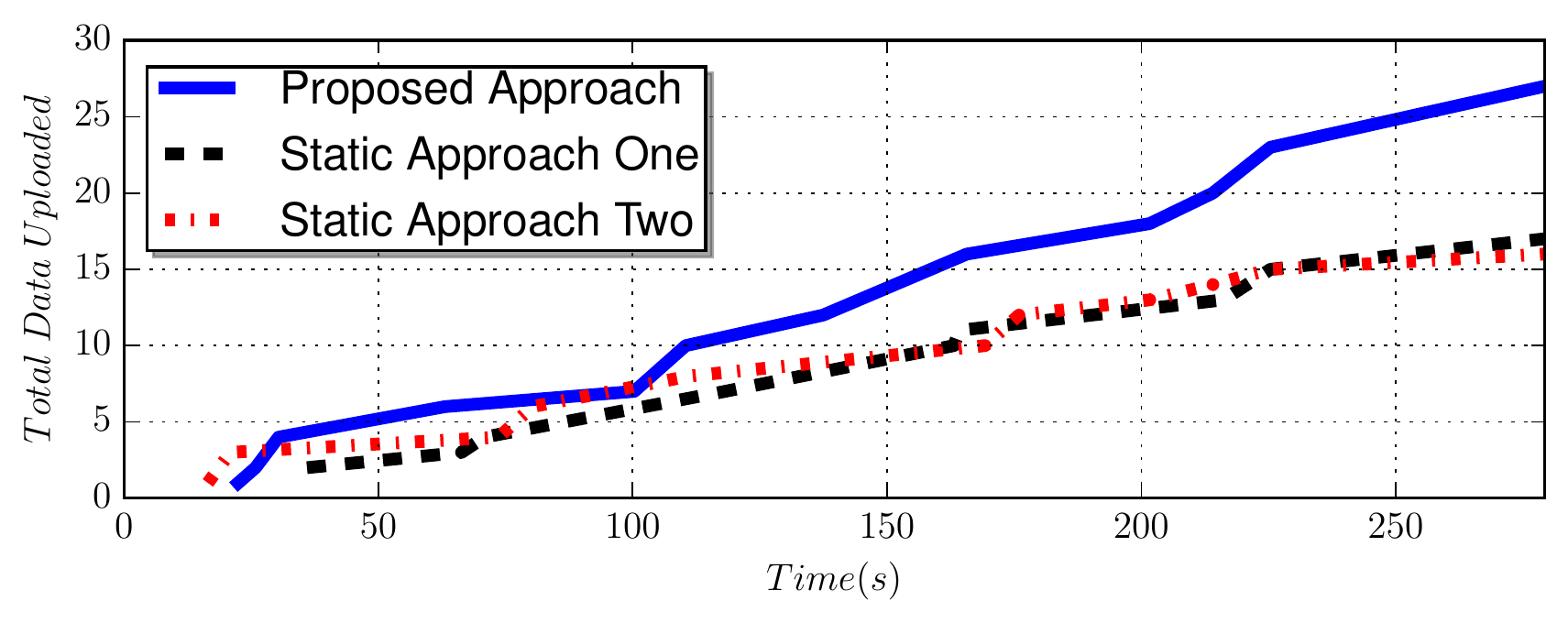}
  \caption{The total amount of data uploaded during the experimental study, under the proposed approach and two static approaches.}
     \label{fig:total-exp}
   \end{figure}

\subsection{Comparison to Static Approaches}\label{sec:compare-exp}

We also compare the performance of our method to the two static approaches introduced in Section~\ref{sec:compare}. The experiment videos for all three cases can be found in~\cite{tro17-videos}.

First, as shown in Figure~\ref{fig:static-ws-exp}, we conducted an experiment using the static approach one for a duration of 3 minutes. Robot~$l_1$ remains still at its initial location for all time, while robots~$a_0,a_1,a_2$ navigate back to robot~$l_1$ once they have gathered enough data that needs to be transferred. As shown in Figure~\ref{fig:total-exp}, $17$ units of data are uploaded in total. 
Second, as shown in Figure~\ref{fig:platoon-ws-exp}, we conducted an experiment using the static approach two, also for a duration of 3 minutes.
The robots form a platoon in the order~$a_2,a_0,l_1,a_1$, so that all source robots~$a_0,a_1,a_2$ are always within the communication range of robot~$l_1$.  Robots~$a_0,a_1,a_2$ take turns to execute their local plans by navigating with the whole group to their desired regions to gather data and transfer the data \emph{directly} to~$l_1$.
As shown in Figure~\ref{fig:total-exp}, $16$ units of data are uploaded in total, compared to $27$ units using the proposed dynamic approach. 

\begin{figure}[t]
\centering
   \includegraphics[width =0.48\textwidth]{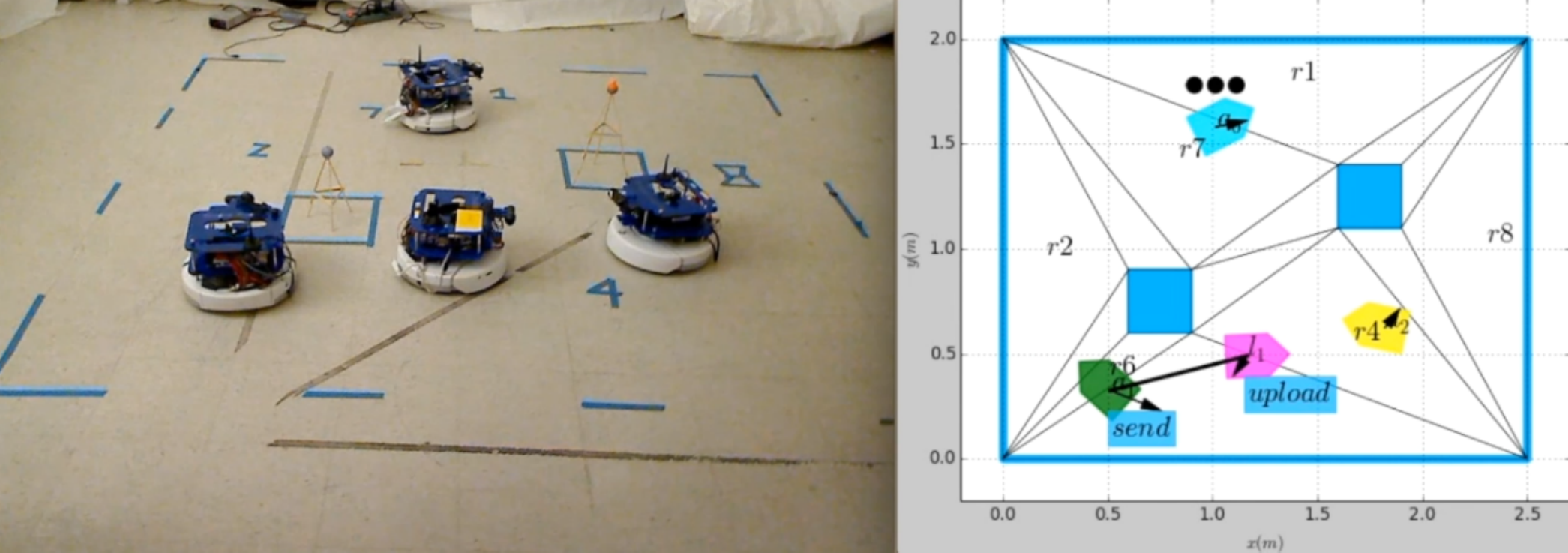}
  \caption{A snapshot for the experiment under the static approach one, where the relay robot~$l_1$ remains static at all time.}
\label{fig:static-ws-exp}
\end{figure}
\begin{figure}[t]
\centering
   \includegraphics[width =0.48\textwidth]{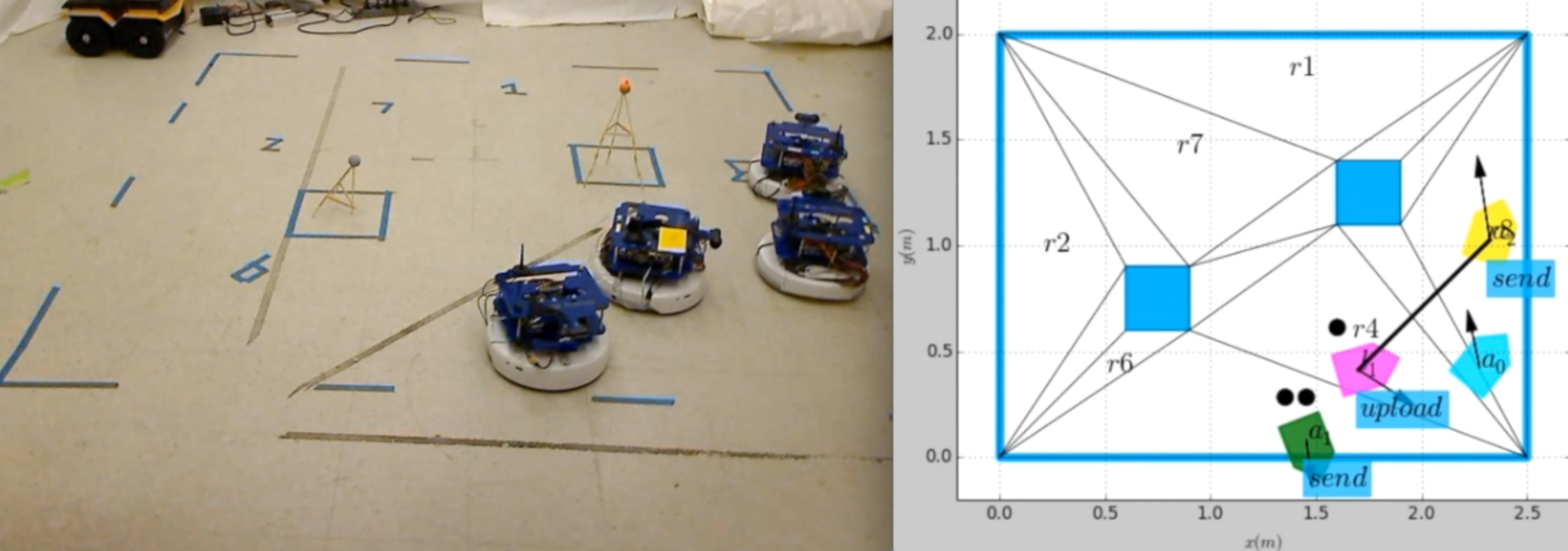}
  \caption{Snapshot of the experiment under the static approach two, where all robots move as a group, being connected at all time.}
\label{fig:platoon-ws-exp}
\end{figure}

Thus similar conclusions can be obtained as in Section~\ref{sec:compare} that the proposed dynamic approach improves greatly the overall data-gathering efficiency compared to the other two static approaches. It is worth mentioning that sequence of \emph{spontaneous} meeting events that happened during the experiment is quite different from the simulated result, due to the inter-robot collision avoidance scheme.

\section{Conclusion and Future Work}\label{sec:future}
In this work we proposed a distributed online framework for multiple  robots that jointly coordinates local data-gathering tasks and intermittent communication events so that the collected data at the robots are transferred to a data center while ensuring that robot buffers do not overflow. 
Unlike most relevant literature that relies on all-time connectivity, the proposed intermittent communication framework allows the robots to operate in disconnect mode and accomplish their tasks free of communication constraints, significantly improving on the performance of data acquisition and delivery. We validated our method through numerical simulations and real experiments, and showed that all local data-gathering tasks are satisfied and the local buffers do not overflow.

\bibliography{meng.bib}
\end{document}